\documentclass[12pt]{article}

\usepackage{amsmath}
\usepackage{graphicx}%
\usepackage{amsfonts}%
\usepackage{amssymb}
\usepackage{mathrsfs}
\usepackage{braket}
\usepackage{dsfont}
\usepackage{amsthm}
\usepackage{enumitem}
\usepackage{relsize}
\usepackage{cite}

\newtheorem{condition}{Condition}
\newtheorem{assumption}{Assumption}
\newtheorem{theorem}{Theorem}

\newtheorem{corollary}{Corollary}

\newtheorem{lemma}{Lemma}

\newtheorem{remark}{Remark}

\title{Error Bounds for Finite-Dimensional Approximations of Input-Output Open Quantum Systems by Subspace Truncation and Adiabatic Elimination\thanks{Research supported by the Australian Research Council}}

\author{Onvaree Techakesari and Hendra I. Nurdin \footnote{O. Techakesari and H. I. Nurdin are with the school of Electrical Engineering and Telecommunications, UNSW Australia, Sydney NSW 2052, Australia. Email: o.techakesari@unsw.edu.au and h.nurdin@unsw.edu.au}}

\begin{document}

\maketitle

\begin{abstract}
An important class of physical systems that are of interest in practice are input-output open quantum systems that can be described by quantum stochastic differential equations and defined on an infinite-dimensional underlying Hilbert space. Most commonly, these systems involve coupling to a quantum harmonic oscillator as a system component. This paper is concerned with error bounds in the finite-dimensional approximations of input-output open quantum systems defined on an infinite-dimensional Hilbert space. We develop a framework for developing error bounds between the time evolution of the state of a class of infinite-dimensional quantum systems and its approximation on a finite-dimensional subspace of the original, when both are initialized in the latter subspace. This framework is then applied to two approaches for obtaining finite-dimensional approximations: subspace truncation and adiabatic elimination.  Applications of the bounds to some physical examples drawn from the literature are provided to illustrate our results.
\end{abstract}


\section{Introduction} \label{sec:intro}

Quantum stochastic differential equations (QSDEs) developed independently by Hudson and Parthasarathy \cite{HP84} and Gardiner and Collett \cite{GC85}  (the latter in a less general form than the former) have been widely used to describe the input-output models of physical open markov quantum systems \cite{WM10,BE08,NJD08}. Such models describe the evolution of Markovian quantum systems interacting with a propagating quantum field, such as a quantum optical field,  and are frequently encountered in quantum optics, optomechanics, and related fields. An example in quantum optics would be a cavity QED (quantum electrodynamics) system where a single atom is trapped inside an optical cavity that interacts with an external coherent laser beam impinging on the optical cavity. These input-output models have subsequently played an important role in the modern  development of quantum filtering and quantum feedback control theory \cite{BvHJ07,BvH08}.  Many types of quantum feedback controllers have been proposed in the literature on the basis of  QSDEs, using both measurement-based quantum feedback control, e.g., \cite{BE08,BvH08,BvHJ07,WM10}, and coherent feedback control, e.g., \cite{JNP06,NJP07b,KNPM10}. Besides,  QSDEs have also been applied in various developments in quantum information processing, such as in quantum computation technology; e.g., see \cite{DK04}.

In various physical systems of interest, one often deals with input-output systems that include coupling to a quantum harmonic oscillator. For instance, typical superconducting circuits that are of interest for quantum information processing consist of artificial two-level atoms coupled to  a transmission line resonator. The former is typically described using a finite-dimensional Hilbert space and the latter is a quantum harmonic oscillator with an infinite-dimensional underlying Hilbert space (i.e., $L^2(\mathbb{R})$, the space of square-integrable complex-valued functions on the real line). Another example is a proposed photonic realization of classical logic based on Kerr  nonlinear optical cavities in \cite{Mab11b}, which is built around a quantum harmonic oscillator with a Kerr nonlinear medium inside it. 
If a mathematical model for such quantum devices is sufficiently simple, it is often possible to simulate the dynamics of the system on a digital computer to assess the predicted performance of the actual device, as carried out in \cite{Mab11b}. The simulation carried out is typically that of a stochastic master equation that simulates the stochastic dynamics of a quantum system when one of its output is observed via laboratory procedures such as homodyne detection or photon counting, see \cite{GZ04,BvHJ07,WM10}. 
However, since it is not possible to faithfully simulate a quantum system with an infinite-dimensional Hilbert space, often in simulations this space is truncated to some finite-dimensional subspace. 
Two approaches that are often employed to approximate a quantum system, defined on an infinite-dimensional space, are subspace truncation approximation and adiabatic elimination (also known as singular perturbation).  Subspace truncation approximation is applied to eliminate higher dimensions of the original infinite-dimensional Hilbert space. An operator $X$ on the infinite-dimensional space is approximated by a truncated operator of the form $PXP$, where $P$ denotes an orthogonal projection projector onto the approximate finite-dimensional subspace. For instance, with quantum harmonic oscillators, a commonly used finite-dimensional space is the span of a finite number of Fock states $|0 \rangle,|1 \rangle, \ldots, |n \rangle$.
On the other hand, adiabatic elimination is often used to simplify quantum systems comprising components that evolve at multiple well-separated time-scales. In this approach, the faster variables are eliminated from the mathematical model description of the systems. 

Despite the ubiquity of approximating infinite-dimensional Hilbert spaces of quantum systems by finite-dimensional subspaces for simulations of input-output quantum systems, to the best of the authors' knowledge, there does not appear to be any work that has tried to obtain some explicit bounds on the approximation error of the joint state of the system and the quantum field it is coupled to. This work develops a framework for  developing bounds on the error between the quantum state of a quantum system  described by the QSDE and the quantum state of a finite-dimensional approximation described by another QSDE, when both systems are initialized in a state in the finite-dimensional subspace. Central to the framework is a contractive semigroup associated with the unitary QSDEs of input-output Markov quantum systems. Error bounds are developed ‎for both adiabatic elimination and subspace truncation approximations. For illustration,  our results are applied to some physical examples drawn from the literature. Prelimary results of this work were announced in the conference paper \cite{TN15b}. The results presented in the work go significantly beyond \cite{TN15b}. In particular, \cite{TN15b} only treates the subspace truncation approximation with some elements of the proofs omitted, error bounds for adiabatic elimination had not been developed, and computability of the error bounds were not considered.  

The rest of this paper is structured as follows. In Section \ref{sec:prelim}, we present the class of open quantum systems and the associated QSDEs describing Markovian open quantum systems. Explicit error bounds for the subspace truncation approximation of a Markovian open quantum system are established in Section \ref{sec:main-truncation} and some examples are provided. We then establish error bounds for adiabatic elimination approximation in Section \ref{sec:main-adiabatic} and some examples are also provided. Finally, concluding remarks close the paper in Section \ref{sec:conclusion}.

\section{Preliminaries} \label{sec:prelim}

\subsection{Notation}

We use $\imath=\sqrt{-1}$ and  let $(\cdot)^*$ denote the adjoint of a linear operator on a Hilbert space as well as the conjugate of a complex number, and $(\cdot)^\top$ denote matrix transposition. We denote by $\delta_{ij}$ the Kronecker delta function. We define $\Re\{A\} = \frac{1}{2} (A + A^*)$ and $\Im\{A\} = \frac{1}{2\imath} (A-A^*)$.
For a linear operator $A$, we write $\ker(A)$ to denote the kernel of $A$ and ${\rm ran}(A)$ the range of $A$. We often write $\ket{\cdot}$ to denote an element of a Hilbert space and denote by ${\cal H}\underline{\otimes}{\cal F}$ the algebraic tensor product of Hilbert spaces ${\cal H}$ and ${\cal F}$.
For a subspace ${\cal H}_0$ of a Hilbert space ${\cal H}$, we write $P_{{\cal H}_0}$ to denote the orthogonal projection operator onto ${\cal H}_0$.
For a Hilbert space ${\cal H} = {\cal H}_0 \oplus {\cal H}_1$, we will write ${\cal H} \ominus {\cal H}_0$ to denote ${\cal H}_1$.  For a linear operator $X$ on $\mathcal{H}$, $\left. X \right|_{\mathcal{H}_0}$ denotes the restriction of $X$ to $\mathcal{H}_0$. 
We use ${\cal B}({\cal H})$ to denote the algebra of all bounded linear operators on ${\cal H}$.
We write $[A,B] = AB - BA$. The notation $\|\cdot\|$ will be used to denote Hilbert space norms and operator norms, $\braket{\cdot,\cdot}$ denotes an inner product on a Hilbert space, linear in the right slot and antilinear in the left, and $\ket{\cdot}\bra{\cdot}$ denotes an outer product. Here, $\mathds{1}_{[0,t]}: [0,t] \rightarrow \{0,1\}$ denotes the indicator function. Finally, $\mathbb{Z}_+$ denotes the set of all positive integers. 

\subsection{Open quantum systems}

Consider a separable Hilbert space ${\cal H}_0$ and the symmetric boson Fock space (of multiplicity $m$) ${\cal F}$ defined over the space $L^2([0,T];\mathbb{C}^m) = \mathbb{C}^m \otimes L^2([0,T])$ with $0 < T < \infty$; see \cite[Ch. 4-5]{Mey95} for more details. 
We will use $e(f) \in {\cal F}$, with $f \in L^2([0,T];\mathbb{C}^m)$, to denote exponential vectors in ${\cal F}$.
Let $\mathfrak{S} \subset L^2([0,T];\mathbb{C}^m) \cap L^\infty_{\rm loc}([0,T];\mathbb{C}^m)$ be an admissible subspace in the sense of Hudson-Parthasarathy \cite{HP84} which contains at least all simple functions, where $L^\infty_{\rm loc}([0,T];\mathbb{C}^m)$ is the space of locally bounded vector-valued functions. Here, we will consider a dense domain ${\cal D}_0 \subset {\cal H}_0$ and a dense domain of exponential vectors ${\cal E} = {\rm span}\{e(f)\mid f \in \mathfrak{S}\} \subset {\cal F}$.

Consider an open Markov quantum system that can be described by a set of linear operators defined on the Hilbert space ${\cal H}_0$: (i) a self-adjoint Hamiltonian operator $H$, (ii) a vector of coupling operators $L$ with the $j$-th element, $L_j: \mathcal{H}_0 \rightarrow {\cal H}_0$ for all $j = 1,2,\ldots,m$, and (iii) a unitary scattering matrix $S$ with the $ij$-th element, $S_{ij}: {\cal H}_0 \rightarrow {\cal H}_0$ for all $i,j = 1,2,\ldots,m$. 
Moreover, the operators $S_{ij},L_j,H$ and their adjoints are assumed to have ${\cal D}_0$ as a common invariant dense domain.
Under this description, we note that $m$ corresponds to the number of external bosonic input fields driving the system.
Each bosonic input field can be described by annihilation and creation field operators, $b^i_t$ and ${b^i_t}^*$, respectively, which satisfy the commutation relations $[b^i_t,{b^j_s}^*] = \delta_{ij}\delta(t-s)$ for all $i,j = 1,2,\ldots,m$ and all $t,s \geq 0$.
We can then define the annihilation process ${\cal A}^i_t$, the creation process ${{\cal A}^i_t}^*$, and the gauge process $\Lambda^{ij}_t$ as
\begin{align*}
 {\cal A}^i_t = \int_0^t b^i_s ds, \qquad
 {{\cal A}^i_t}^* = \int_0^t {b^i_s}^* ds, \qquad
 \Lambda^{ij}_t = \int_0^t {b^i_s}^* b^j_s ds.
\end{align*}
Note that these processes are adapted quantum stochastic processes.
In the vacuum representation, the products of their forward differentials $d{\cal A}^i_t = {\cal A}^i_{t + dt} - {\cal A}^i_t$, 
$d{{\cal A}^i_t}^* = {{\cal A}^i_{t + dt}}^* - {{\cal A}^i_t}^*$, and $d\Lambda^{ij}_t = \Lambda^{ij}_{t + dt} - \Lambda^{ij}_t$ satisfy the quantum It\={o} table
\begin{align*}
\begin{tabular}{c|cccc}
 $\times$ & $d{\cal A}^k_t$ & $d{{\cal A}^k_t}^*$ & $d\Lambda^{k\ell}_t$ & dt \\ \hline
 $d{\cal A}^i_t$ & $0$ & $\delta_{ik} dt$ & $\delta_{ik} d{\cal A}^\ell_t$ & $0$ \\
 $d{{\cal A}^j_t}^*$ & $0$ & $0$ & $0$ & $0$ \\
 $d\Lambda^{ij}_t$ & $0$ & $\delta_{jk} d{{\cal A}^i_t}^*$ & $\delta_{jk} d\Lambda^{i\ell}_t$ & $0$ \\
 $dt$ & $0$ & $0$ & $0$ & $0$
\end{tabular}.
\end{align*}
Here, $b^i_t = \frac{d{\cal A}^i_t}{dt}$ can be interpreted as a vacuum quantum white noise, while $\Lambda^{ii}_t$ can be interpreted as the quantum realization of a Poisson process with zero intensity \cite{HP84}.

Following \cite{BvHS08}, the time evolution of a Markov open quantum system is given by an adapted process $U_t$ satisfying the left Hudson-Parthasarathy QSDE \cite{HP84}:
\begin{align}
 dU_t 
 &= U_t\left\{\sum_{i,j = 1}^m \left( S_{ji}^* - \delta_{ij}) d\Lambda^{ij}_t\right) 
   + \sum_{i = 1}^m \left( L_i^* d{{\cal A}^i_t} \right) \right. \nonumber \\
  &\quad \left. - \sum_{i,j = 1}^m\hspace{-0.3em}\left(S_{ji}^* L_j {d{\cal A}^i_t}^* \right)  + \left[\imath H - \frac{1}{2}\sum_{i=1}^{m}(L_i^*  L_i)\right]dt \right\}
 \label{eq:QSDE}
\end{align}
with $U_0 = I$. The quantum stochastic integrals are defined relative to the domain ${\cal D}_0 \underline{\otimes} {\cal E}$.
With the left QSDE, the evolution of a state vector $\psi \in {\cal H}_0\otimes{\cal F}$ is given by $U_t^* \psi$.

In this paper, we are interested in the problem of approximating the system with operator parameters $(S,L,H)$ by an open quantum system with linear operator parameters $(S^{(k)},L^{(k)},H^{(k)})$ defined on a closed subspace ${\cal H}^{(k)} \subset {\cal H}_0$, where $S^{(k)}$ is unitary and $H^{(k)}$ is self-adjoint. Consider a dense domain ${\cal D}^{(k)} \subset {\cal H}^{(k)}$.
Again, the operators $(S^{(k)},L^{(k)},H^{(k)})$ and their adjoints are assumed to have ${\cal D}^{(k)}$ as a common invariant dense domain.
Similar to \eqref{eq:QSDE}, the time evolution of the approximating system is given by an adapted process $U^{(k)}_t$ satisfying the left Hudson-Parthasarathy QSDE \cite{HP84}:
\begin{align}
 dU^{(k)}_t 
 &= U^{(k)}_t\left\{\sum_{i,j = 1}^m \left( S_{ji}^{(k)*} - \delta_{ij}) d\Lambda^{ij}_t\right) 
 + \sum_{i = 1}^m \left( {L_i^{(k)*}} d{{\cal A}^i_t} \right)  \right. \nonumber \\
  &\quad \left. - \sum_{i,j = 1}^m \left({S^{(k)*}_{ji} L_j}  {d{\cal A}^i_t}^* \right)+ \left[\imath H^{(k)} - \frac{1}{2}\sum_{i=1}^{m}({L^{(k)}_i}^*  L^{(k)}_i)\right]dt \right\},
 \label{eq:QSDE-k}
\end{align}
with $U^{(k)}_0 = I$. Here, the quantum stochastic integrals in the above equation are defined relative to the domain ${\cal D}^{(k)} \underline{\otimes} {\cal E}$.
Similarly, the evolution of a state vector $\psi \in {\cal H}^{(k)}\otimes{\cal F}$ is given by ${U^{(k)}_t}^* \psi$.

\subsection{Associated semigroups}

Let $\theta_t: L^2([t,T];\mathbb{C}^m) \rightarrow L^2([0,T];\mathbb{C}^m)$ be the canonical shift $\theta_t f(s) = f(t+s)$.
We also let $\Theta_t: {\cal F}_{[t} \rightarrow {\cal F}$ denote the second quantization of $\theta_t$, where $\mathcal{F}_{[t}$ denotes the Fock space over $L^2([t,\infty);\mathbb{C}^m)$. Note that an adapted process $U_t$ on ${\cal H}_0\otimes{\cal F}$ is called a contraction (or unitary) cocycle if $U_t$ is a contraction (or unitary) for all $t \geq 0$, $t \mapsto U_t$ is strongly continuous, and $U_{s+t} = U_s(I \otimes \Theta_s^*U_t\Theta_s)$.

Let us now impose an  important condition on the open quantum systems under consideration, adopted from \cite{BvHS08}.

\begin{condition}[Contraction cocycle solutions] \label{cond:cocycle}
For all $t \geq 0$ and all $k \in \mathbb{Z}_+$,
\begin{enumerate}[label=$(\alph*)$] 
 \item \label{cond:cocycle-1} the QSDE \eqref{eq:QSDE} possesses a unique solution $U_t$ which extends to a unitary cocycle on ${\cal H}_0\otimes{\cal F}$,
 \item \label{cond:cocycle-2} the QSDE  \eqref{eq:QSDE-k} possesses a unique solution $U^{(k)}_t$ which extends to a contraction cocycle on ${\cal H}^{(k)}\otimes{\cal F}$.
\end{enumerate}
\end{condition}

Let us define an operator $T_t^{(\alpha\beta)}: {\cal H}_0 \rightarrow {\cal H}_0$ via the identity
\begin{align*}
 \langle u, T^{(\alpha\beta)}_t v\rangle 
 &= e^{-\frac{1}{2}(\|\alpha\|^2 + \|\beta\|^2)t} 
 \left< u \otimes e(\alpha \mathds{1}_{[0,t]}), U_t v \otimes e(\beta \mathds{1}_{[0,t]}) \right>
\end{align*}
for all $u, v \in {\cal H}_0$ and all $\alpha, \beta \in \mathbb{C}^m$.
From \cite[Lemma 1]{BvHS08}, under Condition \ref{cond:cocycle}\ref{cond:cocycle-1}, 
the operator $T_t^{(\alpha\beta)} \in {\cal B}({\cal H}_0)$ is a strongly continuous contraction semigroup on ${\cal H}_0$ and its generator ${\cal L}^{(\alpha\beta)}$ satisfies ${\rm Dom}({\cal L}^{(\alpha\beta)}) \supset {\cal D}_0$ such that
\begin{align}
{\cal L}^{(\alpha\beta)} u 
 &= \left[\sum_{i,j=1}^{m} \left(\alpha_i^* S_{ji}^* \beta_j \right) 
 - \sum_{i,j=1}^{m}\left(\alpha_i^* S_{ji}^* L_j \right)  
 + \sum_{j=1}^{m} \left(L_j ^* \beta_j\right) \right. \nonumber \\
 &\quad \left.  + \left(\imath H - \frac{1}{2}\sum_{i=1}^{m}L^*_i L_i\right) - \frac{\|\alpha\|^2 + \|\beta\|^2}{2}\right] u
 \label{eq:generator}
\end{align}
for all $u \in {\cal D}_0$. Here, we note that ${\rm Dom}({\cal L}^{(\alpha\beta)})$ is dense in ${\cal H}_0$.
We likewise define an operator $T_t^{(k;\alpha\beta)}: {\cal H}^{(k)} \rightarrow {\cal H}^{(k)}$ by replacing $U_t$ with $U_t^{(k)}$.

\begin{condition}[Core for generators] \label{cond:core}
For all $t \geq 0$ and all $k \in \mathbb{Z}_+$,
\begin{enumerate}[label=$(\alph*)$] 
 \item \label{cond:core-1} ${\cal D}_0$ is a core for ${\cal L}^{(\alpha\beta)}$,
 \item \label{cond:core-2} ${\cal D}^{(k)}$ is a core for ${\cal L}^{(k;\alpha\beta)}$.
\end{enumerate}
\end{condition}

Condition \ref{cond:core}\ref{cond:core-1} ensures that the definition \eqref{eq:generator} completely determines ${\cal L}^{(\alpha \beta)}$ for all $\alpha,\beta \in \mathbb{C}^m$, likewise \ref{cond:core}\ref{cond:core-2} completely determines ${\cal L}^{(k;\alpha \beta)}$.
In the sequel, we will make use of the above semigroups associated with open quantum systems in establishing our model approximation error bound. Several sufficient conditions are known to guarantee that a QSDE possesses a unique solution that extends  to a unitary cocycle when the Hilbert space is infinite-dimensional and the operator coefficients of the QSDE  are unbounded, see, e.g., \cite{Fagno90}, and a related discussion in \cite[Remark 4]{BvHS08}.  Throughout the paper, we will assume that Conditions \ref{cond:cocycle} and \ref{cond:core} are fulfilled.

\section{Error bounds for subspace truncation approximations} \label{sec:main-truncation}

In this section, we consider the problem where the infinite-dimensional space ${\cal H}_0$ is truncated to  a finite-dimensional subspace ${\cal H}^{(k)}$, and the original operators $X$ on ${\cal H}_0$ is approximated by truncated operators of the form $P_{{\cal H}^{(k)}}XP_{{\cal H}^{(k)}}$.
Here, the dimension of ${\cal H}^{(k)}$ increases with $k \in \mathbb{Z}_+$ and ${\cal D}^{(k)} = {\cal H}^{(k)}$.
Moreover, Condition \ref{cond:cocycle}\ref{cond:cocycle-2} (in fact, $U^{(k)}_t$ is unitary \cite{HP84}) and \ref{cond:core}\ref{cond:core-2} hold immediately.

\subsection{Assumptions and preliminary results}

\begin{assumption} \label{assump:domains}
For any $k \in \mathbb{Z}_+$ and any $\alpha,\beta \in \mathbb{C}^m$, ${\cal H}^{(k)} \subset {\rm Dom}({\cal L}^{(\alpha\beta)})$.
\end{assumption}

Let ${\cal M}^{(k)} = {\rm ran}\left(\left.\left({\cal L}^{(\alpha,\beta)}-{\cal L}^{(k;\alpha,\beta)}\right)\right|_{{\cal H}^{(k)}}\right)$.
Supposing that Assumption \ref{assump:domains} holds, we also assume the following.

\begin{assumption} \label{assump:approx}
For each $k \in \mathbb{Z}_+$ and each $\alpha,\beta \in \mathbb{C}^m$,
there exists ${\gamma^{(\alpha\beta)}_k}, q^{(k;\alpha\beta)}_{{\cal L}}, q^{(k;\alpha\beta)}_{a}, q^{(k;\alpha\beta)}_{e} >0 $, and a non-trivial subspace $\{0\} \subset {\cal K}^{(k)} \subseteq {\cal H}^{(k)}$  such that
\begin{enumerate}[label=$(\alph*)$] 
 \item \label{assump:subspaceK1} $\left\| \left.\left({\cal L}^{(\alpha\beta)} - {\cal L}^{(k;\alpha\beta)}\right)\right|_{{\cal K}^{(k)}}\right\| \leq q^{(k;\alpha\beta)}_{{\cal L}}$.
 \item \label{assump:subspaceK2} $\left.\left({\cal L}^{(\alpha\beta)} - {\cal L}^{(k;\alpha\beta)}\right)\right|_{{\cal H}^{(k)} \ominus {\cal K}^{(k)}} = 0$, \\ i.e., ${\cal H}^{(k)}\ominus{\cal K}^{(k)} \subseteq \bigcap_{\alpha,\beta \in \mathbb{C}^m} \ker\left(\left.\left({\cal L}^{(\alpha\beta)} - {\cal L}^{(k;\alpha\beta)}\right)\right|_{{\cal H}^{(k)}}\right)$
 \item \label{assump:dissipate-approx}
 For any  $u \in {\cal H}^{(k)}$, 
 \begin{align*}
  &\Re\{\langle P_{{\cal K}^{(k)}} {\cal L}^{(k;\alpha\beta)} u, P_{{\cal K}^{(k)}} u\rangle \} 
  = -g(k,\alpha,\beta) \left\| P_{{\cal K}^{(k)}} u \right\|^2 + h(k,\alpha,\beta,u)
 \end{align*}
 for some $g(k,\alpha,\beta) \geq {\gamma^{(\alpha\beta)}_k}$ and some $\left| h(k,\alpha,\beta,u)\right| \leq q^{(k;\alpha\beta)}_{a} \left\| P_{{\cal K}^{(k)}} u \right\| \|u\|$.
 \item \label{assump:dissipate-true} For any $u \in {\cal H}^{(k)}$ and any $t \geq 0$,  
 \begin{align*}
  &\Re\{\langle T^{(\alpha\beta)}_t {\cal L}^{(\alpha\beta)} P_{{\cal M}^{(k)}}u, T^{(\alpha,\beta)}_t P_{{\cal M}^{(k)}}u\rangle \}
  = -\hat{g}(k,\alpha,\beta) \left\| T^{(\alpha\beta)}_t P_{{\cal M}^{(k)}} u \right\|^2  + \hat{h}(t,k,\alpha,\beta,u)
 \end{align*}
 for some $\hat{g}(k,\alpha,\beta) \geq {\gamma^{(\alpha\beta)}_k}$ and some $\left|\hat{h}(t,k,\alpha,\beta,u)\right| \leq q^{(k;\alpha\beta)}_{e} \left\| T^{(\alpha\beta)}_t P_{{\cal M}^{(k)}} u \right\| \|u\|$.
\end{enumerate}
\end{assumption}

\begin{assumption} \label{assump:convergence}
There exists $r,s \in \mathbb{Z}_+$ such that, for all $\alpha,\beta \in \mathbb{C}^m$, we have that
\begin{align*}
 \lim_{k \rightarrow \infty} q^{(k;\alpha\beta)}_{\cal L} \left( \frac{q^{(k;\alpha\beta)}_{e}}{{\gamma^{(\alpha\beta)}_k}} \right)^{(1-2^{-r})} \left( \frac{q^{(k;\alpha\beta)}_{a}}{{\gamma^{(\alpha\beta)}_k}} \right)^{(1-2^{-s})} = 0
\end{align*}
Moreover, for any $i = 0,1,2,\ldots,r$, and $j = 0,1,2,\ldots,s$, we have that 
\begin{align*}
 &\lim_{k \rightarrow \infty} \frac{q^{(k;\alpha\beta)}_{\cal L}}{{\gamma^{(\alpha\beta)}_k}} \left( \frac{q^{(k;\alpha\beta)}_{e}}{{\gamma^{(\alpha\beta)}_k}} \right)^{(1-2^{-i})} \left( \frac{q^{(k;\alpha\beta)}_{a}}{{\gamma^{(\alpha\beta)}_k}} \right)^{(1-2^{-j})} = 0 \\
 &\lim_{k \rightarrow \infty}  e^{-\gamma_k^{(\alpha \beta)}t}q^{(k;\alpha\beta)}_{\cal L} \left( \frac{q^{(k;\alpha\beta)}_{e}q^{(k;\alpha\beta)}_{a}}{({\gamma^{(\alpha\beta)}_k})^2} \right)^{(1-2^{-\ell})} = 0,
\end{align*}
for any $\alpha,\beta \in \mathbb{C}^m$, any $t>0$, and for $\ell=0,1,2,\ldots,\min\{r,s\}-1$. 
\end{assumption}

Let us present some useful lemmas.

\begin{lemma} \label{lemma:dissipative-approx}
Suppose that Assumption \ref{assump:approx}\ref{assump:dissipate-approx} holds.
Then for any $k,r \in \mathbb{Z}_{+}$, any $\alpha,\beta \in \mathbb{C}^m$, any $u \in {\cal H}^{(k)}$, and any $t \geq 0$, it holds that
\begin{align}
 &\left\| P_{{\cal K}^{(k)}} T^{(k;\alpha\beta)}_t u \right\| 
 \leq \left[ \left(\frac{q^{(k;\alpha\beta)}_{a}}{{\gamma^{(\alpha\beta)}_k}}\right)^{\left(1 - 2^{-r}\right)} 
 + \sum_{j=0}^{r-1} c_j e^{-\left(2^{-j}\right){\gamma^{(\alpha\beta)}_k} t} \left(\frac{q^{(k;\alpha\beta)}_{a}}{{\gamma^{(\alpha\beta)}_k}}\right)^{\left(1 - 2^{-j}\right)}\right] \|u\|
 \label{eq:dissipative-approx}
\end{align}
where $c_0 = 1$ and $c_j = \sqrt{c_{j-1} 2^j\left( 2^j -1 \right)^{-1}}$ for $j \geq 1$.
\end{lemma}
\begin{proof}
First note, from the definition of a strongly continuous semigroup, that $T_0^{(\alpha\beta)} = I$ and $\frac{d}{dt} T_t^{(\alpha\beta)} u = {\cal L}^{(\alpha\beta)} u$ \cite{CZ95}.
From Assumption \ref{assump:approx}\ref{assump:dissipate-approx}, we have that
\begin{align*}
 \frac{d}{dt}\left\| P_{{\cal K}^{(k)}} T^{(k;\alpha\beta)}_t u \right\|^2 
 &= \frac{d}{dt} \left< P_{{\cal K}^{(k)}} T^{(k;\alpha\beta)}_t u, P_{{\cal K}^{(k)}} T^{(k;\alpha\beta)}_t u \right> \nonumber \\
 &= 2\Re\left\{ \left< P_{{\cal K}^{(k)}} {\cal L}^{(k;\alpha\beta)} T^{(k;\alpha\beta)}_t u,  P_{{\cal K}^{(k)}} T^{(k;\alpha\beta)}_t u\right> \right\} \nonumber \\
 &= \hspace{-0.3em} -2g(k,\alpha,\beta) \left\| P_{{\cal K}^{(k)}} T^{(k;\alpha\beta)}_t u \right\|^2 \hspace{-0.3em} + 2h(k,\alpha,\beta, T^{(k;\alpha\beta)}_t u).
\end{align*}
Solving the above ODE gives us that
\begin{align}
 \left\| P_{{\cal K}^{(k)}} T^{(k;\alpha\beta)}_t u \right\|^2 
 &= e^{-2g(k,\alpha,\beta)t} \left\| P_{{\cal K}^{(k)}} u \right\|^2 + 2\int_0^t h(k,\alpha,\beta, T^{(k;\alpha,\beta)}_\tau u) e^{-2g(k,\alpha,\beta)(t-\tau)} d\tau \nonumber \\
 &\leq e^{-2{\gamma^{(\alpha\beta)}_k} t} \left\| P_{{\cal K}^{(k)}} u \right\|^2  + 2q^{(k;\alpha\beta)}_{a} \int_0^t e^{-2{\gamma^{(\alpha\beta)}_k}(t-\tau)} \left\| P_{{\cal K}^{(k)}} T^{(k;\alpha\beta)}_\tau u \right\| \|u\| d\tau \label{eq:dissipate-proof-1} \\
 &\leq \left[e^{-2{\gamma^{(\alpha\beta)}_k} t} + \frac{q^{(k;\alpha\beta)}_{a}}{{\gamma^{(\alpha\beta)}_k}} \right] \|u\|^2. \nonumber
\end{align}
Here, the second step follows from Assumption \ref{assump:approx}\ref{assump:dissipate-approx}. 
The last step follows because $\left\| P_{{\cal K}^{(k)}} u \right\| \leq \|u\|$ and $T^{(k;\alpha\beta)}_t$ is a contraction semigroup.
Noticing that $\sqrt{|a|^2 + |b|^2} \leq |a|+|b|$ for any $a,b \in \mathbb{R}$, we have that
\begin{align}
 \left\| P_{{\cal K}^{(k)}} T^{(k;\alpha\beta)}_t u \right\| 
 &\leq \left[e^{-{\gamma^{(\alpha\beta)}_k} t} + \left(\frac{q^{(k;\alpha\beta)}_{a}}{{\gamma^{(\alpha\beta)}_k}}\right)^{\frac{1}{2}} \right] \|u\|.
 \label{eq:dissipate-proof-2}
\end{align}
Now, substituting \eqref{eq:dissipate-proof-2} into the right-handed side of \eqref{eq:dissipate-proof-1}, we have that
\begin{align*}
 &\left\| P_{{\cal K}^{(k)}} T^{(k;\alpha\beta)}_t u \right\|^2 \nonumber \\
 &\leq e^{-2{\gamma^{(\alpha\beta)}_k} t} \|u\|^2 + 2q^{(k;\alpha\beta)}_{a} \int_0^t e^{-{\gamma^{(\alpha\beta)}_k}(2t - \tau)}\|u\|^2 d\tau \nonumber \\
 &\quad + 2q^{(k;\alpha\beta)}_{a} \int_0^t e^{-2{\gamma^{(\alpha\beta)}_k}(t-\tau)} \left(\frac{q^{(k;\alpha\beta)}_{a}}{{\gamma^{(\alpha\beta)}_k}}\right)^{\frac{1}{2}}  \|u\|^2 d\tau \nonumber \\
 &= \left[e^{-2{\gamma^{(\alpha\beta)}_k} t} + 2e^{-{\gamma^{(\alpha\beta)}_k} t}\left(1 - e^{-{\gamma^{(\alpha\beta)}_k} t}\right)\left(\frac{q^{(k;\alpha\beta)}_{a}}{{\gamma^{(\alpha\beta)}_k}}\right) + \left(1 - e^{-2{\gamma^{(\alpha\beta)}_k} t}\right)\left(\frac{q^{(k;\alpha\beta)}_{a}}{{\gamma^{(\alpha\beta)}_k}}\right)^{\frac{3}{2}} \right] \|u\|^2 \nonumber \\
 &\leq \left[\hspace{-0.2em} e^{-2{\gamma^{(\alpha\beta)}_k} t} \hspace{-0.2em} + \hspace{-0.2em} 2e^{-{\gamma^{(\alpha\beta)}_k} t}\left(\frac{q^{(k;\alpha\beta)}_{a}}{{\gamma^{(\alpha\beta)}_k}}\right) \hspace{-0.2em} + \hspace{-0.2em} \left(\frac{q^{(k;\alpha\beta)}_{a}}{{\gamma^{(\alpha\beta)}_k}}\right)^{\frac{3}{2}} \hspace{-0.2em} \right] \hspace{-0.3em} \|u\|^2.
\end{align*}
Taking the square root on both sides of the equation, we get
\begin{align*}
 \left\| P_{{\cal K}^{(k)}} T^{(k;\alpha\beta)}_t u \right\| 
 &\leq \left[ e^{-{\gamma^{(\alpha\beta)}_k} t} + \sqrt{2}e^{-\frac{1}{2}{\gamma^{(\alpha\beta)}_k} t}\left(\frac{q^{(k;\alpha\beta)}_{a}}{{\gamma^{(\alpha\beta)}_k}}\right)^{\hspace{-0.2em}\frac{1}{2}} + \left(\frac{q^{(k;\alpha\beta)}_{a}}{{\gamma^{(\alpha\beta)}_k}}\right)^{\frac{3}{4}} \right]  \|u\|.
\end{align*}
From repeat application of the above steps, we establish the lemma statement.
\end{proof}


\begin{lemma} \label{lemma:dissipative-true}
Suppose that Assumption \ref{assump:approx}\ref{assump:dissipate-true} holds.
Then for any $k,r \in \mathbb{Z}_+$, any $\alpha,\beta \in \mathbb{C}^m$, any $u \in {\cal H}^{(k)}$, and any $t \geq 0$, it holds that
\begin{align}
 \left\| T^{(\alpha\beta)}_t P_{{\cal M}^{(k)}} u \right\| 
 \leq \left[ \left(\frac{q^{(k;\alpha\beta)}_{e}}{{\gamma^{(\alpha\beta)}_k}}\right)^{\left(1 - 2^{-r}\right)} + \sum_{j=0}^{r-1} c_j e^{-\left(2^{-j}\right){\gamma^{(\alpha\beta)}_k} t} \left(\frac{q^{(k;\alpha\beta)}_{e}}{{\gamma^{(\alpha\beta)}_k}}\right)^{\left(1 - 2^{-j}\right)}\right] \|u\|
 \label{eq:dissipative-true}
\end{align}
where $c_0 = 1$ and $c_j = \sqrt{c_{j-1} 2^j\left( 2^j -1 \right)^{-1}}$ for $j \geq 1$.
\end{lemma}
\begin{proof}
Similar to Lemma \ref{lemma:dissipative-approx}, using Assumption \ref{assump:approx}\ref{assump:dissipate-true}, we have that
\begin{align*}
 \frac{d}{dt}\left\| T^{(\alpha\beta)}_t P_{{\cal M}^{(k)}} u \right\|^2 
 &= 2\Re\{ \braket{T^{(\alpha\beta)}_t {\cal L}^{(\alpha\beta)} P_{{\cal M}^{(k)}} u, T^{(\alpha\beta)}_t P_{{\cal M}^{(k)}} u} \} \nonumber \\
 &= -2\hat{g}(k,\alpha,\beta) \left\| T^{(\alpha\beta)}_t P_{{\cal M}^{(k)}} u \right\|^2 + 2\hat{h}(t,k,\alpha,\beta,u).
\end{align*}
Solving the ODE, we have that
\begin{align*}
 \left\| T^{(\alpha\beta)}_t P_{{\cal M}^{(k)}} u \right\|^2 
 &\leq e^{-2{\gamma^{(\alpha\beta)}_k} t} \left\| P_{{\cal M}^{(k)}} u \right\|^2 + 2q^{(k;\alpha\beta)}_{e} \int_0^t e^{-2{\gamma^{(\alpha\beta)}_k}(t-\tau)} \left\| T^{(\alpha\beta)}_t P_{{\cal M}^{(k)}}u \right\| \|u\| d\tau \nonumber \\
 &\leq \left[e^{-2{\gamma^{(\alpha\beta)}_k} t} + \frac{q^{(k;\alpha\beta)}_{e}}{{\gamma^{(\alpha\beta)}_k}} \right] \|u\|^2.
\end{align*}
The lemma statement is then established by following similar arguments to Lemma \ref{lemma:dissipative-approx}.
\end{proof}


\subsection{Error bounds for finite-dimensional approximations}

We begin by defining
\begin{align}
 &z^{k}_{r,s}(t,\alpha,\beta) \nonumber \\
 &:= q^{(k;\alpha\beta)}_{\cal L} \left[ 
 t\left(\frac{q^{(k;\alpha\beta)}_{e}}{{\gamma^{(\alpha\beta)}_k}}\right)^{(1-2^{-r})}  \left(\frac{q^{(k;\alpha\beta)}_{a}}{{\gamma^{(\alpha\beta)}_k}}\right)^{(1-2^{-s})} \right. \nonumber \\
 &\quad+ \sum_{i=0}^{r-1}  \frac{2^i c_i}{{\gamma^{(\alpha\beta)}_k}} \left(1 - e^{-2^{-i}{\gamma^{(\alpha\beta)}_k} t}\right) \hspace{-0.2em} \left(\frac{q^{(k;\alpha\beta)}_{e}}{{\gamma^{(\alpha\beta)}_k}}\right)^{(1-2^{-i})} \hspace{-0.5em} \left(\frac{q^{(k;\alpha\beta)}_{a}}{{\gamma^{(\alpha\beta)}_k}}\right)^{(1-2^{-s})} \nonumber \\
 &\quad+ \sum_{i=1}^{s-1} \frac{2^i c_i}{{\gamma^{(\alpha\beta)}_k}} \left(1 - e^{-2^{-i}{\gamma^{(\alpha\beta)}_k} t}\right) \hspace{-0.2em} \left(\frac{q^{(k;\alpha\beta)}_{a}}{{\gamma^{(\alpha\beta)}_k}}\right)^{(1-2^{-i})} \hspace{-0.5em} \left(\frac{q^{(k;\alpha\beta)}_{e}}{{\gamma^{(\alpha\beta)}_k}}\right)^{(1-2^{-r})} \nonumber \\
 &\quad+ \sum_{i=0}^{r-1} \sum_{\stackrel{j=0}{j \neq i}}^{s-1} c_i c_j \frac{2^{(i+j)}}{\left(2^i - 2^j\right){\gamma^{(\alpha\beta)}_k}} \left( e^{-2^{-i}{\gamma^{(\alpha\beta)}_k} t } - e^{-2^{-j}{\gamma^{(\alpha\beta)}_k} t } \right) \hspace{-0.2em}
  \left(\frac{q^{(k;\alpha\beta)}_{e}}{{\gamma^{(\alpha\beta)}_k}}\right)^{\hspace{-0.5em}(1-2^{-i})} \left(\frac{q^{(k;\alpha\beta)}_{a}}{{\gamma^{(\alpha\beta)}_k}}\right)^{\hspace{-0.5em}(1-2^{-j})} \nonumber \\
 &\quad+ \left. t \sum_{i=0}^{\min\{r,s\}-1} c_i^2 e^{-2^{-i} {\gamma^{(\alpha\beta)}_k} t} \left(\frac{q^{(k;\alpha\beta)}_{e}q^{(k;\alpha\beta)}_{a}}{({\gamma^{(\alpha\beta)}_k})^2}\right)^{(1-2^{-i})} \right],
  \label{eq:bound}
\end{align}
where $c_0 = 1$, and $c_j = \sqrt{c_{j-1}  2^j(2^j-1)^{-1}}$ for $j  \geq 1$.
We now establish an error bound between the two semigroups associated with the open quantum systems.

\begin{lemma} \label{lem:semigroup-error}
Suppose Assumptions \ref{assump:domains} and \ref{assump:approx} hold.
Then for any $k,r,s \in \mathbb{Z}_+$, any $\alpha,\beta \in \mathbb{C}^m$, any $u \in {\cal H}^{(k)}$, and any $t \geq 0$,  it holds that
\begin{align}
 \left\|\left(T^{(\alpha\beta)}_{t} - T^{(k;\alpha\beta)}_{t}\right)  u \right\| 
 &\leq z^{k}_{r,s}(t,\alpha,\beta)  \|u\|.
 \label{eq:semigroup-error}
\end{align}
\end{lemma}
\begin{proof}
First note, from the definition of a strongly continuous semigroup, that \cite{CZ95}
\begin{enumerate} \renewcommand{\theenumi}{\roman{enumi}}
 \item  $T_0^{(\alpha\beta)} = T_0^{(k;\alpha\beta)} = I$ for all $k \in \mathbb{Z}_+$,
 \item $\frac{d}{dt} T_t^{(\alpha\beta)} u = {\cal L}^{(\alpha\beta)} u$  for all $u \in {\rm Dom}({\cal L}^{(\alpha\beta)})$,
 \item $\frac{d}{dt} T_t^{(k;\alpha\beta)} u = {\cal L}^{(k;\alpha\beta)} u$ for all $u \in {\cal H}^{(k)}$ (since ${\cal H}^{(k)}$ is finite-dimensional).
\end{enumerate}
From the above properties and Assumption \ref{assump:domains}, we can write for all $u \in {\cal H}^{(k)}$ and all $t \geq 0$ that
\begin{align}
 \frac{d}{dt} \left(T^{(\alpha\beta)}_{t} - T^{(k;\alpha\beta)}_{t}\right) u 
 &= \left( {\cal L}^{(\alpha\beta)}T^{(\alpha\beta)}_{t} - {\cal L}^{(k;\alpha\beta)}T^{(k;\alpha\beta)}_{t}\right) u \nonumber \\
 &= {\cal L}^{(\alpha\beta)}\left(T^{(\alpha\beta)}_{t} - T^{(k;\alpha\beta)}_{t}\right) u + \left. \left( {\cal L}^{(\alpha\beta)} - {\cal L}^{(k;\alpha\beta)}\right) \right|_{\mathcal{H}^{(k)}} T^{(k;\alpha\beta)}_{t} u
 \label{eq:cauchy-problem}
\end{align}
with $\left(T^{(\alpha\beta)}_0 - T^{(k;\alpha\beta)}_0\right) u = 0$. Note that by Assumption \ref{assump:approx}, $\left. \left( {\cal L}^{(\alpha\beta)} - {\cal L}^{(k;\alpha\beta)}\right) \right|_{\mathcal{H}^{(k)}}$ is a bounded operator.
Since ${\cal L}^{(\alpha\beta)}$ is a generator of a semigroup on a Hilbert space, \\
$$
\left(T^{(\alpha\beta)}_0 - T^{(k;\alpha\beta)}_0\right) u \in {\rm Dom}({\cal L}^{(\alpha\beta)}),$$ 
(due to Assumption \ref{assump:domains}), 
and 
$$
\left. \left( {\cal L}^{(\alpha\beta)} - {\cal L}^{(k;\alpha\beta)}\right) \right|_{\mathcal{H}^{(k)}} T^{(k;\alpha\beta)}_{t} u \in {\cal C}^{1}([0,t];{\cal H}_0)$$ 
(the class of continuously differentiable functions from $[0,t]$ to ${\cal H}_0$), a unique solution of \eqref{eq:cauchy-problem} exists and is given by \cite[Thm 3.1.3]{CZ95}
\begin{align*}
 \left(T^{(\alpha\beta)}_{t} - T^{(k;\alpha\beta)}_{t}\right) u = \int_0^{t} T^{(\alpha\beta)}_{{t}-\tau} \left( {\cal L}^{(\alpha\beta)} - {\cal L}^{(k;\alpha\beta)}\right) T^{(k;\alpha\beta)}_{\tau} u d\tau
\end{align*}
for all $t \geq 0$ and all $u \in {\cal H}^{(k)}$.
From Assumption \ref{assump:approx}\ref{assump:subspaceK2} and the definition of ${\cal M}^{(k)}$, we then have for all $u \in {\cal H}^{(k)}$ and all $t \geq 0$ that
\begin{align}
 \left\|\left(T^{(\alpha\beta)}_{t} - T^{(k;\alpha\beta)}_{t}\right) u \right\| 
 &\leq \int_0^{t} \left\|T^{(\alpha\beta)}_{{t}-\tau} \left( {\cal L}^{(\alpha\beta)} - {\cal L}^{(k;\alpha\beta)}\right) T^{(k;\alpha\beta)}_{\tau} u \right\| d\tau \nonumber \\
 &= \int_0^{t} \left\|T^{(\alpha\beta)}_{{t}-\tau} P_{{\cal M}^{(k)}} \left( {\cal L}^{(\alpha\beta)} - {\cal L}^{(k;\alpha\beta)}\right) P_{{\cal K}^{(k)}} T^{(k;\alpha\beta)}_{\tau} u \right\| \hspace{-0.2em} d\tau.
 \label{eq:semigroup-proof1}
\end{align}

Now using the bounds \eqref{eq:dissipative-approx} and \eqref{eq:dissipative-true} (established in Lemmas \ref{lemma:dissipative-approx} and \ref{lemma:dissipative-true}, respectively), and applying Assumption \ref{assump:approx}\ref{assump:subspaceK1}, we have that
\begin{align*}
 &\left\|T^{(\alpha\beta)}_{{t}-\tau} P_{{\cal M}^{(k)}} \left( {\cal L}^{(\alpha\beta)} - {\cal L}^{(k;\alpha\beta)}\right) P_{{\cal K}^{(k)}} T^{(k;\alpha\beta)}_{\tau} u \right\|  \nonumber \\
 &\leq q^{(k;\alpha\beta)}_{\cal L}\left[ \left(\frac{q^{(k;\alpha\beta)}_{e}}{{\gamma^{(\alpha\beta)}_k}}\right)^{\hspace{-0.2em}(1-2^{-r})} \hspace{-0.3em} \left(\frac{q^{(k;\alpha\beta)}_{a}}{{\gamma^{(\alpha\beta)}_k}}\right)^{\hspace{-0.2em}(1-2^{-s})} \right. \nonumber \\
 &\quad + \left(\frac{q^{(k;\alpha\beta)}_{a}}{{\gamma^{(\alpha\beta)}_k}}\right)^{\hspace{-0.3em}(1-2^{-s})} \hspace{-0.3em} \Bigg[ \sum_{i=0}^{r-1}  c_i e^{-2^{-i}{\gamma^{(\alpha\beta)}_k} (t-\tau)} \left(\frac{q^{(k;\alpha\beta)}_{e}}{{\gamma^{(\alpha\beta)}_k}}\right)^{\hspace{-0.3em}(1-2^{-i})} \hspace{-0.2em}\Bigg] \nonumber \\ 
 &\quad + \left(\frac{q^{(k;\alpha\beta)}_{e}}{{\gamma^{(\alpha\beta)}_k}}\right)^{\hspace{-0.2em}(1-2^{-r})} 
 \Bigg[ \sum_{i=1}^{s-1} c_i e^{-2^{-i}{\gamma^{(\alpha\beta)}_k} \tau} \left(\frac{q^{(k;\alpha\beta)}_{a}}{{\gamma^{(\alpha\beta)}_k}}\right)^{(1-2^{-i})} \Bigg] \nonumber \\
 &\quad + \left. \Bigg[ \sum_{i=0}^{r-1} \sum_{j=0}^{s-1} c_i c_j e^{{\gamma^{(\alpha\beta)}_k}\left(-2^{-i} t + (2^{-i} - 2^{-j}) \tau\right)}
 \left(\frac{q^{(k;\alpha\beta)}_{e}}{{\gamma^{(\alpha\beta)}_k}}\right)^{(1-2^{-i})}\left(\frac{q^{(k;\alpha\beta)}_{a}}{{\gamma^{(\alpha\beta)}_k}}\right)^{(1-2^{-j})}  \Bigg]\right] \|u\|.
\end{align*}
The result \eqref{eq:semigroup-error} then follows from substitution of the above identity into \eqref{eq:semigroup-proof1} and integration. This establishes the lemma statement.
\end{proof}


Let $\mathfrak{S}' \subset L^2([0,T];\mathbb{C}^m)$ denote the dense set of all simple functions  in $L^2([0,T];\mathbb{C}^m)$.
That is, for any $t \in [0,T]$ and $f \in \mathfrak{S}'$, there exists $0 < \ell < \infty $ and a sequence $0 = t_0 < t_1 < \cdots < t_{\ell} < t_{\ell+1} = t$ such that $f= \sum_{i=0}^{\ell} \alpha(i) \mathds{1}_{[t_i,t_{i+1})}$ for some constants $\alpha(i) \in \mathbb{C}^m$,   $i = 0,1,\ldots,\ell$. 
Let ${\cal U}^{(k)} = \{ u\otimes e(f)\mid u \in {\cal H}^{(k)}, f \in \mathfrak{S}' \}$. We can now proceed to derive error bounds for approximations by subspace truncation.

\begin{lemma} \label{lemma:main1}
Suppose Assumptions \ref{assump:domains} and \ref{assump:approx} hold. For any $\psi_1 = u_1 \otimes e(f_1),\psi_2 = u_2 \otimes e(f_2) \in {\cal U}^{(k)}$, let $t_0=0 <t_1 < \ldots <t_{\ell}<t_{\ell+1}=t$ with $0<t \leq T$ be a sequence such that $f_j= \sum_{i=0}^{\ell} \alpha_j(i) \mathds{1}_{[t_i,t_{i+1})}$ for all $i = 0,1,\ldots,\ell$, with $\alpha_j(i) \in \mathbb{C}^m$ for $j=1,2$ and $i=0,1,\ldots,\ell$. Then for any $k,r,s \in \mathbb{Z}_+$,   
\begin{align}
 \left| \left< \left(U_t - {U^{(k)}_t}\right)^* \psi_1, \psi_2 \right> \right| 
 &\leq \sum_{i=0}^{\ell} z^{k}_{r,s}\left((t_{i+1} - t_{i}),\alpha_1(i), \alpha_2(i)\right) \|\psi_1\| \|\psi_2\|.
 \label{eq:Ut-err1}
\end{align}
\end{lemma}
\begin{proof}
First recall that our admissible subspace $\mathfrak{S}$ contains $\mathfrak{S}'$. Hence, ${\cal U}^{(k)} \subset {\cal H}^{(k)} \underline{\otimes} {\cal E}$ and the quantum stochastic integrals are well defined for all $\psi \in {\cal U}^{(k)}$. Also, recall that $\|\psi\|^2 = \braket{u\otimes e(f),  u\otimes e(f)} = \|u\|^2 \|e(f)\|^2$ for any $\psi \in {\cal H}_0\otimes {\cal F}$ \cite{Mey95}. 
Using the cocycle properties (Condition \ref{cond:cocycle}) as well as the definitions of $T^{(\alpha\beta)}_t$ and $T^{(k;\alpha\beta)}_t$,  we have the identity \cite{BvHS08}
\begin{align}
 & \left<U_t^* \psi_1, \psi_2 \right>= \left\|e(f_1)\right\|  \left\|e(f_2)\right\| \langle u_1, T_{t_1-t_0}^{(\alpha_1(0)\alpha_2(0))} \cdots T_{t-t_{\ell}}^{(\alpha_1(\ell)\alpha_2(\ell))}  u_2 \rangle, \label{eq:weak-bound-id} 
\end{align}
and likewise when $U_t$ and $T^{(\alpha \beta)}_t$ are respectively replaced by $U_t^{(k)}$ and $T^{(k;\alpha \beta)}_t$, from which we immediately obtain
\begin{align}
 &\left| \left<\left(U_t - U_t^{(k)}\right)^* \psi_1, \psi_2 \right>\right| \nonumber \\
 &\leq \left\|\psi_1\right\| \left\|e(f_2)\right\| \left\|\left(T_{t_1-t_0}^{(\alpha_1(0)\alpha_2(0))} \cdots T_{t-t_{\ell}}^{(\alpha_1(\ell)\alpha_2(\ell))}  - T_{t_1-t_0}^{(k;\alpha_1(0)\alpha_2(0))} \cdots T_{t-t_{\ell}}^{(k;\alpha_1(\ell)\alpha_2(\ell))}\right) u_2\right\|.
 \label{eq:weak-bound} 
\end{align}
Now, note that for any $u \in {\cal H}^{(k)}$ that
\begin{align*}
 &\left(T_{t_1-t_0}^{(\alpha_1(0)\alpha_2(0))}T_{t_2-t_1}^{(\alpha_1(1)\alpha_2(1))} \cdots T_{t-t_{\ell}}^{(\alpha_1(\ell)\alpha_2(\ell))} - T_{t_1-t_0}^{(k;\alpha_1(0)\alpha_2(0))} T_{t_2-t_1}^{(k;\alpha_1(1)\alpha_2(1))} \cdots T_{t-t_{\ell}}^{(k;\alpha_1(\ell)\alpha_2(\ell))}\right) u \nonumber \\
 &= \left[ \left( T_{t_1-t_0}^{(\alpha_1(0)\alpha_2(0))} -  T_{t_1-t_0}^{(k;\alpha_1(0)\alpha_2(0))} \right) \right.  T_{t_2-t_1}^{(k;\alpha_1(1)\alpha_2(1))} T_{t_3-t_2}^{(k;\alpha_1(2)\alpha_2(2))} \cdots T_{t-t_{\ell}}^{(k;\alpha_1(\ell)\alpha_2(\ell))} \nonumber \\
 &\quad +  T_{t_1-t_0}^{(\alpha_1(0)\alpha_2(0))} \left( T_{t_2-t_1}^{(\alpha_1(1)\alpha_2(1))} -  T_{t_2-t_1}^{(k;\alpha_1(1)\alpha_2(1))}  \right)  T_{t_2-t_1}^{(k;\alpha_1(1)\alpha_2(1))} \cdots T_{t-t_{\ell}}^{(k;\alpha_1(\ell)\alpha_2(\ell))} \nonumber \\
 &\quad \left. + \cdots + T_{t_1-t_0}^{(\alpha_1(0)\alpha_2(0))} \cdots T_{t_{\ell}-t_{\ell-1}}^{(\alpha_1(\ell-1)\alpha_2(\ell-1))}  \left( T_{t-t_\ell}^{(\alpha_1(\ell)\alpha_2(\ell))} -  T_{t-t_\ell}^{(k;\alpha_1(\ell)\alpha_2(\ell))}  \right) \right] u.
\end{align*}
From \eqref{eq:weak-bound}, the bound \eqref{eq:semigroup-error} (established in Lemma \ref{lem:semigroup-error}), and that fact that the semigroups are contractions, we have for any $u \in {\cal H}^{(k)}$ that
\begin{align}
 &\left\|\left(T_{t_1-t_0}^{(\alpha_1(0)\alpha_2(0))}T_{t_2-t_1}^{(\alpha_1(1)\alpha_2(1))} \cdots T_{t-t_{\ell}}^{(\alpha_1(\ell)\alpha_2(\ell))} - T_{t_1-t_0}^{(k;\alpha_1(0)\alpha_2(0))} T_{t_2-t_1}^{(k;\alpha_1(1)\alpha_2(1))} \cdots T_{t-t_{\ell}}^{(k;\alpha_1(\ell)\alpha_2(\ell))}\right) \hspace{-0.2em} u\right\| \nonumber \\
 &\leq \sum_{i=0}^\ell z^{k}_{r,s}\left((t_{i+1}-t_i),\alpha_1(i),\alpha_2(i)\right) \|u\|.
 \label{eq:semigroups-bound}
\end{align}
The bound \eqref{eq:Ut-err1} then follows by substituting \eqref{eq:semigroups-bound} into \eqref{eq:weak-bound}. This establishes the theorem statement.
\end{proof}


\begin{corollary}\label{cor:main2}
Suppose Assumptions \ref{assump:domains}, \ref{assump:approx}, and \ref{assump:convergence} hold. 
For any $t \in [0,T]$ with $0 < T < \infty$, any $\psi_1 = u_1\otimes e(f_1),\psi_2 = u_2\otimes e(f_2) \in {\cal H}^{(k)}\otimes{\cal F}$,  
we have that
\begin{align}
 \left| \left< \left(U_t - U^{(k)}_t\right)^* \psi_1,\psi_2 \right> \right| 
 &\leq 2 \big( \|u_1\| \|e(f_1)-e(f'_1)\| \|\psi_2\| + \|u_2\| \|e(f_2)-e(f'_2)\| \|\psi_1\| \big) \nonumber \\
 &\quad + \sum_{i=0}^{\ell}  z^{k}_{r,s}\left((t_{i+1} - t_{i}),f'_{1,i}, f'_{2,i}\right) \|\psi_1'\|  \|\psi_2'\|.
 \label{eq:Ut-err2}
\end{align}
for any $\psi_j' = u_j  \otimes e(f_j') \in {\cal U}^{(k)}$ with $f'_j  =\sum_{i=0}^{\ell} f'_{j,i} \mathds{1}_{[t_i,t_{i+1})}$ for some $\ell \in \mathbb{Z}_+$,  some sequence $t_0=0 <t_1 < \ldots <t_{\ell}<t_{\ell+1}=t$, and some constants $f'_{j,i} \in \mathbb{C}^m$ for $j=1,2$ and $i=0,1,\ldots,\ell$. Moreover, for any fixed positive integer $p \in \mathbb{Z}_+$, and any $\psi_1=u_1\otimes e(f_1), \psi_2=u_2\otimes e(f_2) \in {\cal H}^{(p)}\otimes{\cal F}$,
\begin{eqnarray}
 \lim_{k\rightarrow\infty} \left| \left< \left(U_t - U^{(k)}_t\right)^* \psi_1,\psi_2 \right> \right| &=& 0.
 \label{eq:weak-convergence}
\end{eqnarray}
\end{corollary}
\begin{proof}
Recall that $U_t$ is  unitary  and $U^{(k)}_t$ is a contraction (Condition \ref{cond:cocycle}).
From the  triangle inequality and Cauchy-Schwarz's inequality, we note that
\begin{align*}
 \left| \left<\left(U_t - {U^{(k)}_t}\right)^* \psi_1, \psi_2 \right>  \right| 
 &\leq \left| \left< \left(U_t - {U^{(k)}_t}\right)^*  u_1 \otimes \left(e(f_1) - e(f'_1)\right), u_2\otimes e(f_2) \right> \right| \nonumber \\
 &\quad + \left| \left< \left(U_t - {U^{(k)}_t}\right)^* u_1\otimes e(f'_1), u_2 \otimes \left(e(f_2) - e(f'_2)\right) \right> \right| \nonumber \\
 &\quad + \left| \left< \left(U_t - {U^{(k)}_t}\right)^* u_1\otimes e(f'_1), u_2 \otimes e(f'_2) \right> \right| \nonumber \\
 &\leq 2 \|u_1\| \left\| e(f_1) - e(f'_1) \right\| \|\psi_2\| + 2 \|u_2\| \left\| e(f_2) - e(f'_2) \right\| \left\| \psi'_1 \right\| \nonumber \\
 &\quad + \left| \left< \left(U_t - {U^{(k)}_t}\right)^* u_1\otimes e(f'_1), u_2 \otimes e(f'_2) \right> \right|.
\end{align*}
The result \eqref{eq:Ut-err2} then follows from the bound \eqref{eq:Ut-err1} (established in Lemma \ref{lemma:main1}).

To show \eqref{eq:weak-convergence}, recall that $\mathfrak{S}'$ is dense in $L^2([0,T];\mathbb{C}^m)$. Therefore, for each $\epsilon > 0$ and each $f_j \in L^2([0,T];\mathbb{C}^m)$, there exists $f'_j \in \mathfrak{S}'$ such that $\|f_j - f_j'\| < \epsilon$.
Moreover,  there exists $0 < \ell < \infty$ and $0 = t_0 < t_1 < \cdots < t_\ell < t_{\ell+1} = t$, such that $f'_j = \sum_{i=0}^{\ell} \alpha'_j(i)\mathds{1}_{[t_i,t_{i+1})}$ with $\alpha'_j(0),\alpha'_j(1),\ldots,\alpha'_j(\ell) \in \mathbb{C}^m$. Suppose that $u_1,u_2 \neq 0$ (otherwise the corollary statement becomes trivial), then for any $\epsilon > 0$ we may choose $f'_1,f'_2 \in \mathfrak{S}'$ (choosing $f_1'
$ first followed by $f_2'$) such that 
\begin{align*}
 \left\| e(f_1) - e(f'_1) \right\| &< \frac{\epsilon}{6\|u_1\| \left\| e(f_2) \right\|} \\
 \left\| e(f_2) - e(f'_2) \right\| &< \frac{\epsilon}{6\|u_2\| \left\| e(f'_1) \right\|}
\end{align*}
Finally, from Assumption \ref{assump:convergence} and the bound \eqref{eq:Ut-err1} (established in Lemma \ref{lemma:main1}), we can find a sufficiently large $k \in \mathbb{Z}_+$, larger than $p$, such that 
\begin{align*}
 \sum_{i=0}^\ell z^{k}_{r,s}\left(((t_{i+1}-t_i),\alpha'_1(i),\alpha'_2(i)\right))  < \frac{\epsilon}{3 \| \psi_1'\| \|\psi_2'\|}. 
\end{align*}
From \eqref{eq:Ut-err2} and the above choices, we then have that
\begin{align*}
 \left| \left<\left(U_t - {U^{(k)}_t}\right)^* \psi_1, \psi_2 \right>  \right| < \epsilon.
\end{align*}
This establishes the corollary statement.
\end{proof}


\begin{theorem} \label{thm:main3}
Suppose that Assumptions \ref{assump:domains},\ref{assump:approx}, and \ref{assump:convergence} hold. 
Let $0 < T < \infty$. For any $t \in [0,T]$, consider any $L'_t \in \mathbb{Z}_+$,
any $\psi'_t = \sum_{j=1}^{L'_t}  \psi'_{j,t}$,  with $\psi'_{j,t} = u'_{j,t} \otimes e(g'_{j,t}) \in {\cal U}^{(k)}$  and  $u'_{j,t} \neq 0$. Also, consider any $f' \in \mathfrak{S}'$. Let $\ell$ be a positive integer and $t_0 =0 <t_1 < \ldots < t_{\ell} <t_{\ell +1}=t$ be a sequence such that $f' =\sum_{i=0}^{\ell} f'_{i} \mathds{1}_{[t_i,t_{i+1})}$ and $g'_{j,t} =\sum_{i=0}^{\ell} g'_{j,t,i} \mathds{1}_{[t_i,t_{i+1})}$ for some constants $f'_{i},g_{j,t,i} \in \mathbb{C}^m$ for for $j=1,2,\ldots,L'_t$ and $i=0,1,\ldots,\ell$. Let $u \in {\cal H}^{(k)}$ with $\| u\|=1$, and $|f\rangle  = e(f)/\|e(f)\| \in {\cal F}$  (i.e., $|f\rangle$ is a coherent state with amplitude $f$), and $\psi = u \otimes |f \rangle$. Then,
\begin{align}
\left\|\left(U_t - U^{(k)}_t\right)^* \psi \right\|^2  &\leq 4 \Big( \left\| |f \rangle - |f' \rangle \right\|  + \left\| {U_t}^* \psi - \psi'_t \right\| \Big) \nonumber \\
 &\quad + 2  \sum_{j=1}^{L'_t} \sum_{i=1}^\ell z^{k}_{r,s}\left((t_{i+1} - t_{i}),f'_i, g'_{j,t,i}\right)\| \psi'_{j,t}\|. 
 \label{eq:Ut-err3}
\end{align}
If $U_t^{(k)}$ is unitary for each $t$ then the following bound holds,
\begin{align}
\left\|\left(U_t - U^{(k)}_t\right)^* \psi \right\|^2  &\leq 4 \Big( \left\| |f \rangle - |f' \rangle \right\|  + \left\| {U_t^{(k)}}^* \psi - \psi'_t \right\| \Big) \nonumber \\
 &\quad + 2  \sum_{j=1}^{L'_t} \sum_{i=1}^\ell z^{k}_{r,s}\left((t_{i+1} - t_{i}),f'_i, g'_{j,t,i}\right)\| \psi'_{j,t}\|. 
 \label{eq:Ut-err3-unitary}
\end{align}
Moreover, for any fixed $p \in \mathbb{Z}_+$ and any $\psi = u\otimes |f \rangle  \in {\cal H}^{(p)}\otimes{\cal F}$,
\begin{align}
 \lim_{k\rightarrow\infty} \left\|\left(U_t - U^{(k)}_t\right)^* \psi \right\| = 0.
 \label{eq:strong-convergence}
\end{align}
for any $t \in [0,T]$ with $0 < T < \infty$.
\end{theorem}
\begin{remark}
Note that a stronger result of strong convergence uniformly over compact time intervals,  $\lim_{k\rightarrow\infty} \mathop{\sup}_{0 \leq t \leq T} \left\|\left(U_t - U^{(k)}_t\right)^* \psi \right\| =0$, has been established in \cite[Proposition 20]{BvHS08} employing a Trotter-Kato theorem. However, no error bound as in \eqref{eq:Ut-err3} for a finite value of $k$ has previously been established.
\end{remark}

\begin{proof}
First note that since $U_t$ is unitary and $U_t^{(k)}$ a contraction,
\begin{align*}
 \left\|\left(U_t - U^{(k)}_t\right)^* \psi \right\|^2 
 &\leq  2\Re\{\langle (U_t-U_t^{(k)})^* \psi,  U_t^* \psi \rangle \}\\
 &\leq 2 |\langle (U_t-U_t^{(k)})^* \psi,  U_t^* \psi \rangle|.
 \end{align*}
Also, we have that 
\begin{align*}
 \left| \left< \left(U_t - U^{(k)}_t\right)^* \psi, U_t^*\psi \right> \right| 
 &\leq \left| \left< \left(U_t - U^{(k)}_t\right)^* u \otimes \left( |f \rangle  - |f' \rangle \right), U_t^* \psi \right> \right| \nonumber \\
 &\quad + \left| \left< \left(U_t - U^{(k)}_t\right)^* u \otimes |f'\rangle, \left(U_t^*\psi - \psi'_t\right) \right> \right| \nonumber \\
 &\quad + \left| \left< \left(U_t - U^{(k)}_t\right)^* u \otimes | f' \rangle, \psi'_t \right> \right| \nonumber \\
 &\leq 2\|u\| \|f \rangle -|f' \rangle \| \|\psi\| + 2\|u \| \| |f' \rangle\| \left\| U_t^* \psi - \psi'_t \right\| \nonumber \\
 &\quad + \sum_{j=1}^{L'_t} \left|\left<\left(U_t - U^{(k)}_t\right)^* u \otimes | f' \rangle, \psi'_{j,t} \right>\right|.
\end{align*}
The result \eqref{eq:Ut-err3} then follows from the bound \eqref{eq:Ut-err1}, and substituting $\|u \|=1$, $\|\psi\|=1$, and $\| |f' \rangle\|=1$. If $U_t^{(k)}$ is unitary  then we have the bound,
\begin{align*}
 \left\|\left(U_t - U^{(k)}_t\right)^* \psi \right\|^2 
 &= 2\Re\{\langle (U_t^{(k)}-U_t)^* \psi,  {U_t^{(k)}}^* \psi \rangle \}\\
 &\leq 2 |\langle (U_t-U_t^{(k)})^* \psi,  {U_t^{(k)}}^* \psi \rangle|,
 \end{align*}
and following analogous calculations to the above yields,
\begin{align*}
 \left| \left< \left(U_t - U^{(k)}_t\right)^* \psi, {U_t^{(k)}}^*\psi \right> \right| 
  &\leq 2\|u\| \|f \rangle -|f' \rangle \| \|\psi\| + 2\|u \| \| |f' \rangle\| \left\| {U_t^{(k)}}^* \psi - \psi'_t \right\| \nonumber \\
 &\quad + \sum_{j=1}^{L'_t} \left|\left<\left(U_t - U^{(k)}_t\right)^* u \otimes | f' \rangle, \psi'_{j,t} \right>\right|,
\end{align*}
leading to the alternative bound \eqref{eq:Ut-err3-unitary}.

To show \eqref{eq:strong-convergence}, let $t \in [0,T]$ be fixed. Suppose $\psi \neq 0$. Then for any $\epsilon > 0$, we may choose $f' \in \mathfrak{S}'$ such that
\begin{align*}
 \left\| |f \rangle - |f' \rangle \right\| &< \frac{\epsilon^2}{12}.
\end{align*}
Since $T$ is finite, for $k_0 \in \mathbb{Z}_+$ sufficiently large we can choose $0 < L'_t < \infty$ and $\psi'_{j,t} \in {\cal U}^{(k_0)}$ for $j = 1,2,\ldots,L'_t$ such that 
\begin{align*}
 \left\| {U}_t^* \psi - \psi'_t \right\| &< \frac{\epsilon^2}{12},
\end{align*}
with $ \psi'_t=\sum_{j=0}^{L'_t} \psi'_{j,t}$. Finally, from Corollary \ref{cor:main2}, since $L'_t$ is finite we can find $k_1 \in \mathbb{Z}_+$ sufficiently large,  with $k_1>\max\{p, k_0\}$, such that
\begin{align*}
 \left| \left< \left(U_t - U^{(k)}_t\right)^* u \otimes |f' \rangle, \psi'_t \right> \right| &< \frac{\epsilon^2}{6}.
\end{align*}
for all $k>k_1$. From the above choices and taking square roots on both sides of \eqref{eq:Ut-err3}, we have that
\begin{align*}
 \left\|\left(U_t - U^{(k)}_t\right)^* \psi \right\| < \epsilon,
\end{align*}
for all $k>k_1$. Since the theorem statement holds trivially when $\psi=0$, this completes the proof.
\end{proof}

A discussion of the error bounds presented in Theorem \ref{thm:main3} is now in order, starting with \eqref{eq:Ut-err3}. The error bound on the right hand side of \eqref{eq:Ut-err3} is the sum of three terms. The first term  is a bound on the error committed by approximating $e(f)$ by  $e(f')$ for some simple function $f'$. The second term bounds the error in approximating $U_t^*\psi$ by a  finite sum of terms in $\mathcal{U}^{(k)}$ given in $\psi'_t$. Finally, the last term gives an upper bound for the magnitude of the inner product between the error term  $(U_t-U_t^{(k)})^*u \otimes e(f')$ and $\psi'_t$. The first and final terms are computable. However, the second term is difficult.  We note that
\begin{eqnarray}
\left\| U_t^* \psi -\sum_{j=1}^{L'_t} \psi'_{j,t} \right\|&=& \left(\left\langle  U_t^* \psi -\sum_{j=1}^{L'_t} \psi'_{j,t}, U_t^* \psi -\sum_{j=1}^{L'_t} \psi'_{j,t} \right\rangle \right)^{1/2}, \notag \\
&=& \left(\|\psi \|^2 - 2 \sum_{j}^{L'_t} \Re\{\langle  U_t^* \psi, \psi'_{j,t} \rangle \} + \left\| \sum_{j=1}^{L'_t} \psi'_{j,t} \right\|^2 \right)^{1/2}, \notag \\
&=&  \left( \|\psi \|^2 - 2 \sum_{j}^{L'_t} \| e(g'_{j,t}) \| \Re\left\{\left\langle  u, T^{(f'_0 g'_{j,t,0})}_{t_1-t_0} T^{(f'_1 g'_{j,t,1})}_{t_2-t_1} \cdots \right. \right. \right.  \notag \\
&&\;  \left. \left. \left. T^{(f'_{\ell} g'_{j,t,\ell})}_{t_{\ell+1}-t_{\ell}}  u'_{j,t} \right\rangle \right\} + \left\| \sum_{j=1}^{L'_t} \psi'_{j,t} \right\|^2 \right)^{1/2}, 
\label{eq:approx-QSDE-sol-0}
\end{eqnarray}
where in the last line, we again use the identity \eqref{eq:weak-bound-id}. However, this is difficult to compute as it involves the semigroup $T^{(\alpha \beta)}_t$ which acts  on the infinite dimensional space $\mathcal{H}_0$. Thus to alleviate this difficulty we now  turn to the alternative bound  \eqref{eq:Ut-err3-unitary}. 

The first and third terms of \eqref{eq:Ut-err3-unitary} are the same as for \eqref{eq:Ut-err3}. However, for the second term we have the identity 
\begin{eqnarray}
\left\| {U_t^{(k)}}^* \psi -\sum_{j=1}^{L'_t} \psi'_{j,t} \right\| 
&=&  \left( \|\psi \|^2 - 2 \sum_{j}^{L'_t} \| e(g'_{j,t}) \| \Re\left\{\left\langle  u, T^{(k;f'_0 g'_{j,t,0})}_{t_1-t_0} T^{(k;f'_1 g'_{j,t,1})}_{t_2-t_1} \cdots \right. \right. \right.  \notag \\
&&\;  \left. \left. \left. T^{(\ell;f'_{\ell} g'_{j,t,\ell})}_{t_{\ell+1}-t_{\ell}}  u'_{j,t} \right\rangle \right\} + \left\| \sum_{j=1}^{L'_t} \psi'_{j,t} \right\|^2 \right)^{1/2}, 
\label{eq:approx-QSDE-sol}
\end{eqnarray}
derived in the same manner as \eqref{eq:approx-QSDE-sol-0}. However, unlike \eqref{eq:approx-QSDE-sol-0}, \eqref{eq:approx-QSDE-sol} involves only the semigroup $T^{(k;\alpha \beta)}$ which acts on the finite dimensional Hilbert space $\mathcal{H}^{(k)}$. Thus it is a quantity that will be much easier to compute. All that remains now is to construct a suitable approximation $\psi'_{j,t}$ to ${U_t^{(k)}}^* \psi$. One way to do this is to choose $\psi'_{j,t}$ to locally minimize the right hand side of \eqref{eq:approx-QSDE-sol} or, equivalently, the term under the square root. Unfortunately, although this is in principle possible, it is in general a challenging and computationally intensive non-convex optimization problem. We will demonstrate this optimization in an example that will follow.
 

\subsection{Subspace truncation examples} \label{sec:ex}

\subsubsection{Kerr-nonlinear optical cavity:} \label{ex:mabuchi}

Consider a single-mode Kerr-nonlinear optical cavity coupled to a single external coherent field ($m=1$), which is used in the construction of the photonic logic gates presented in \cite{Mab11b}. Let ${\cal H}_0 = \ell^2$ (the space of infinite complex-valued sequences with $\sum_{n=1}^\infty |x_n|^2 < \infty$) which has an orthonormal Fock state basis $\{\ket{n}\}_{n\geq 0}$. 
On this basis $\{\ket{n}\}_{n\geq 0}$, the annihilation, creation, and number operators of the cavity oscillator can be defined (see, e.g., \cite{BvHS08}) satisfying 
\begin{align*}
 a\ket{n} = \sqrt{n}\ket{n-1}, \qquad a^*\ket{n} = \sqrt{n+1}\ket{n+1}, \qquad a^*a\ket{n} = n\ket{n},
\end{align*}
respectively. Similar to Examples 14-15 of \cite{BvHS08}, we set ${\cal D}_0 = {\rm span}\{ \ket{n} \mid n \in \mathbb{Z}_+ \}$.
The Kerr-nonlinear optical cavity can be described by 
\begin{align*}
 S=I, \qquad L = \sqrt{\lambda}a, \qquad H = \Delta a^*a + \chi a^*a^*aa,
\end{align*}
where $\lambda,\Delta,\chi > 0$.
We now show that Conditions \ref{cond:cocycle}\ref{cond:cocycle-1} and \ref{cond:core}\ref{cond:core-1} hold for the Kerr-nonlinear cavity.

Consider ${\cal H}^{(k)} = {\rm span}\left\{ \ket{n} \mid n = 0,1,2,\ldots,k \right\}$ and a system approximation of the form
\begin{align}
 S^{(k)} = I, \qquad L^{(k)} = P_{{\cal H}^{(k)}} L P_{{\cal H}^{(k)}}, \qquad H^{(k)} = P_{{\cal H}^{(k)}} H P_{{\cal H}^{(k)}}.
 \label{eq:approx-system}
\end{align}
Conditions \ref{cond:cocycle}\ref{cond:cocycle-2} and \ref{cond:core}\ref{cond:core-2} hold immediately because ${\cal D}^{(k)} = {\cal H}^{(k)}$ is finite dimensional. Note that ${\cal H}^{(k)} \subset {\cal D}_0$. Hence, Assumption \ref{assump:domains} holds. Recall that ${\cal M}^{(k)} = {\rm ran}\left(\left.\left({\cal L}^{(\alpha,\beta)}-{\cal L}^{(k;\alpha,\beta)}\right)\right|_{{\cal H}^{(k)}}\right)$. From the $(S,L,H)$, we see that ${\cal M}^{(k)} = {\rm span}\{ \ket{k+1} \}$. 
Now consider ${\cal K}^{(k)} = {\rm span}\left\{ \ket{k} \right\}$ and ${\gamma^{(\alpha\beta)}_k} = \frac{1}{2}\left(\lambda k + |\alpha-\beta|^2\right)$.
We will now show that Assumptions \ref{assump:approx} and \ref{assump:convergence} hold for the Kerr-nonlinear optical cavity and the approximation.

\paragraph*{Assumption \ref{assump:approx}\ref{assump:subspaceK1}:}
Note that
\begin{align*}
 \left.\left({\cal L}^{(\alpha\beta)} - {\cal L}^{(k;\alpha\beta)}\right)\right|_{{\cal K}^{(k)}} \hspace{-0.5em}
 &= \beta \sqrt{\lambda} \left. a^* \right|_{{\cal K}^{(k)}}.
\end{align*}
Thus, we have that Assumption \ref{assump:approx}\ref{assump:subspaceK1} holds with $q^{(k;\alpha\beta)}_{\cal L} = \sqrt{\lambda(k+1)}|\beta|$.

\paragraph*{Assumption \ref{assump:approx}\ref{assump:subspaceK2}:}
This assumption follows for the defined ${\cal K}^{(k)}$ because \\ $\ker\left(\left.\left({\cal L}^{(\alpha\beta)} - {\cal L}^{(k;\alpha\beta)}\right)\right|_{{\cal H}^{(k)}} \right) = {\rm span}\left\{ \ket{n} \mid n = 0,1,2,\ldots,k-1 \right\}$ for all $\alpha,\beta \in \mathbb{C}$.

\paragraph*{Assumption \ref{assump:approx}\ref{assump:dissipate-approx}:}
First note that ${L^{(k)}}^* P_{{\cal K}^{(k)}} = 0$,
$P_{{\cal K}^{(k)}} L^{(k)} P_{{\cal H}^{(k)}} = 0$, and \\
$P_{{\cal K}^{(k)}}\left(\frac{1}{2}{L^{(k)}}^* L^{(k)} - \imath H\right) P_{{\cal H}^{(k)}\ominus{\cal K}^{(k)}} = 0$.
Also, for any $u \in {\cal H}^{(k)}$,
$P_{{\cal K}^{(k)}}{L^{(k)}}^* L^{(k)} P_{{\cal K}^{(k)}}u =  \lambda k P_{{\cal K}^{(k)}}u$.
From these identities and the fact that $H^{(k)}$ is self-adjoint, we have for any $u \in {\cal H}^{(k)}$ that
\begin{align*}
 \Re\{\langle P_{{\cal K}^{(k)}} {\cal L}^{(k;\alpha\beta)} u, P_{{\cal K}^{(k)}} u\rangle \} 
 &= \Re\{\langle P_{{\cal K}^{(k)}} {\cal L}^{(k;\alpha\beta)} \left(P_{{\cal K}^{(k)}} + P_{{\cal H}^{(k)}\ominus{\cal K}^{(k)}}\right) u, P_{{\cal K}^{(k)}} u\rangle \} \nonumber \\
 &= -\frac{1}{2} \left(\lambda k + |\alpha|^2 + |\beta|^2 - 2\Re\{\alpha^*\beta\} \right) \left\| P_{{\cal K}^{(k)}} u \right\|^2 \nonumber \\
 &\quad + \Re\left\{\left< \sqrt{\lambda} \beta P_{{\cal K}^{(k)}} a^* P_{{\cal H}^{(k)}\ominus{\cal K}^{(k)}} u, P_{{\cal K}^{(k)}} u\right> \right\} \nonumber \\
 &= -g(k,\alpha,\beta) \left\| P_{{\cal K}^{(k)}} u \right\|^2 + h(k,\alpha,\beta,u).
\end{align*}
Notice that $|\alpha|^2 + |\beta|^2 - 2\Re\{\alpha^*\beta\} =  |\alpha-\beta|^2$. Also, note that 
\begin{align*}
 \bigg| \Re\left\{\left< \sqrt{\lambda} \beta P_{{\cal K}^{(k)}} a^* P_{{\cal H}^{(k)}\ominus{\cal K}^{(k)}} u, P_{{\cal K}^{(k)}} u\right> \right\}\bigg| 
 &\leq |\beta|\sqrt{\lambda} \left\| P_{{\cal K}^{(k)}} a^* P_{{\cal H}^{(k)}\ominus{\cal K}^{(k)}} u\right\| \left\| P_{{\cal K}^{(k)}} u \right\| \nonumber \\
 &\leq |\beta| \sqrt{\lambda k} \|u\|  \left\| P_{{\cal K}^{(k)}} u \right\|.
\end{align*}
Therefore, we have that Assumption \ref{assump:approx}\ref{assump:dissipate-approx} holds for the defined ${\cal K}^{(k)}$and the defined ${\gamma^{(\alpha\beta)}_k}$ with $q^{(k;\alpha\beta)}_{a} =  |\beta| \sqrt{\lambda k}$.

\paragraph*{Assumption \ref{assump:approx}\ref{assump:dissipate-true}:}
For any $u \in {\cal H}_0$, $H P_{{\cal M}^{(k)}}u = \left(\Delta (k+1) + \chi k(k+1) \right) P_{{\cal M}^{(k)}}u$. 
This implies that $\Re\{ \langle T^{(\alpha\beta)}_t (\imath H) P_{{\cal M}^{(k)}} u, T^{(\alpha\beta)}_t P_{{\cal M}^{(k)}} u \rangle \} = 0$. Also, ${L}^* L P_{{\cal M}^{(k)}}u = \lambda(k+1) P_{{\cal M}^{(k)}}u$ for any $u \in {\cal H}_0$.
Therefore, we have that
\begin{align*}
 & \Re\{\langle T^{(\alpha\beta)}_t {\cal L}^{(\alpha\beta)} P_{{\cal M}^{(k)}} u, T^{(\alpha\beta)}_t P_{{\cal M}^{(k)}} u\rangle \} \nonumber \\
 &=-\frac{1}{2} \big(\lambda(k+1) + |\alpha|^2 + |\beta|^2 - 2\Re\{\alpha^*\beta\}  \bigg) \left\| T^{(\alpha\beta)}_t P_{{\cal M}^{(k)}} u \right\|^2 \nonumber \\
 &\quad + \Re\left\{\left< T^{(\alpha\beta)}_t \sqrt{\lambda}\left(\beta a^* - \alpha^*a \right) P_{{\cal M}^{(k)}} u, T^{(\alpha\beta)}_t P_{{\cal M}^{(k)}} u\right> \right\} \nonumber \\ 
 &= -\hat{g}(k,\alpha,\beta) \left\| T^{(\alpha\beta)}_t P_{{\cal M}^{(k)}} u \right\|^2 + \hat{h}(t,k,\alpha,\beta,u).
\end{align*}
Similarly to the previous derivation, using that $T^{(\alpha\beta)}_t$ is a contraction, we see that Assumption \ref{assump:approx}\ref{assump:dissipate-true} holds for the defined ${\cal K}^{(k)}$and the defined ${\gamma^{(\alpha\beta)}_k}$ with
$q^{(k;\alpha\beta)}_{e} = |\alpha|\sqrt{\lambda(k+1)} + |\beta|\sqrt{\lambda(k+2)}$.

\paragraph*{Assumption \ref{assump:convergence}:}
From the defined ${\gamma^{(\alpha\beta)}_k}$, $q^{(k;\alpha\beta)}_{\cal L}$, $q^{(k;\alpha\beta)}_{a}$, $q^{(k;\alpha\beta)}_{e}$, we see that this assumption holds for any $r,s \in \mathbb{Z}_+$ such that $r+s \geq 3$. 

Finally, because Assumptions \ref{assump:domains}-\ref{assump:convergence} hold, Lemma \ref{lemma:main1}, Corollary \ref{cor:main2}, and Theorem \ref{thm:main3} can be applied to obtain error bounds on the finite-dimensional approximations. 

\paragraph{Numerical example of the Kerr-nonlinear optical cavity:} 
Consider $\lambda = 25$, $\Delta = 50$, and $\chi = -\Delta/60$ (these parameters are used in \cite{Mab11b}).
We also consider an input field with a constant amplitude of $\alpha = 0.1$ for $t \in [0,T]$ with $T = 5$.
Let $\psi = \ket{0}\otimes | \alpha\mathds{1}_{[0,t]} \rangle$.
We will now compute an error bound on $\left|\left| \left(U_t - U^{(k)}_t \right)^*\psi \right|\right|$ for different values of $k$.
From \eqref{eq:Ut-err3} (established in Theorem \ref{thm:main3}) and the fact that $\alpha\mathds{1}_{[0,t]} \in \mathfrak{S}'$ is a simple function, we have that
\begin{align*}
 \left\|\left(U_t - U^{(k)}_t\right)^* \psi \right\|^2 
  &\leq4  \left\| {U_t^{(k)}}^* \psi - \psi'_t \right\|  \nonumber \\
  &\quad + 2 \sum_{j=1}^{L'_t} \sum_{i=1}^\ell z^{k}_{r,s}\left((t_{i+1} - t_{i}),f'\mathds{1}_{[t_i,t_{i+1})}, g'_{j,t}\mathds{1}_{[t_i,t_{i+1})}\right] \|\psi'_{j,t}\|
\end{align*}
where $f' = \alpha\mathds{1}_{[0,t]}$, $\psi'_t = \sum_{j=1}^{L'_t} \psi'_{j,t}$, and $\psi'_{j,t} = u'_{j,t}\otimes e(g'_{j,t})$ with $g'_{j,t} \in \mathfrak{S}'$.
Now, using  \eqref{eq:approx-QSDE-sol} we have
\begin{align}
 \left\| {U_t^{(k)}}^* \psi - \psi'_t \right\|^2
 &=\| \psi \|^2 + \| \psi'_t \|^2 \nonumber \\
 &\quad  - 2   \sum_{j=1}^{L'_t} \| e(g'_{j,t}) \| \Re\left\{ \left< 0 \left| T^{(k;\alpha,\beta'_j(0))}_{t_1-t_0}  T^{(k;\alpha,\beta'_j(1))}_{t_2-t_1} \cdots T^{(\alpha,\beta'_j(\ell))}_{t-t_\ell} u'_{j,t} \right.\right> \right\} \label{eq:sim-bound}
\end{align}
where $\beta'_j(i) = g'_{j,t}\mathds{1}_{[t_{i},t_{i+1})}(s)$ for any $s \in [t_i,t_{i+1})$.
To find an appropriate $\psi'_t$, we set up the cost function
\begin{align}
 J_k(\psi'_t) 
 &= \left(1+\| \psi'_t \|^2  - 2  \sum_{j=1}^{L_t} \| e(g'_{j,t}) \| \Re\left\{ \left< 0 \left| T^{(k;\alpha,\beta'_j(0))}_{t_1-t_0}  T^{(k;\alpha,\beta'_j(1))}_{t_2-t_1} \cdots T^{(k;\alpha,\beta'_j(\ell))}_{t-t_\ell} u'_{j,t} \right.\right> \right\} \right)^{1/2},
 \label{eq:sim-cost}
 \end{align}
for any $\psi'_t  = \sum_{j=1}^{L'_t} u'_{j,t}\otimes e(g'_{j,t})$. We then choose $u'_{j,t}$ and $g'_{j,t}$ for $j=1,2,\ldots,L'_t$ such that it is a local minimizer of $J_k$.

For computational simplicity, let us fix $L_t = 1$ and take $t_{i+1}-t_i = 0.5$ for all $i = 0,1,2,\ldots,\ell$. 
With $t= t_{\ell+1}=5$, we then have that $\ell = 9$. We take an initialize guess at $\psi'_{1,t} = \ket{0}\otimes e(\alpha \mathds{1}_{[0,T]})$. 
Using the general purpose unconstrained optimization function \texttt{fminunc} in Matlab, a local minimizer $\psi'_t = u'_{1,t}\otimes e(g'_{1,t})$ was found as $u'_{1,t} = 0.9999\ket{0} - (0.0024 - 0.0094\imath)\ket{1} -0.0001\ket{2}$ and $g'_{1,t} = (0.0866 + 0.0462\imath) \mathds{1}_{[t_0,t_1)} + (0.0882+0.0471\imath)\mathds{1}_{[t_1,T)}$.
The corresponding cost is $J( \psi'_t) = 0.0096$. Consider $r = s = 2$.
Using $\psi'_t$, error bounds on $\left|\left| \left(U_t - U^{(k)}_t \right)^*\psi \right|\right|$ for various values of $k$ are shown in Table \ref{tab:spa-results}. Recall that the dimension of the reduced subspace is $k+1$.

\begin{table}[!t]
 \centering
  \caption{Numerical computation of error bounds on $\left|\left| \left(U_t - U^{(k)}_t \right)^*\psi \right|\right|$ for the Kerr-nonlinear optical cavity}
 \begin{tabular}{c||c}
  k & Error Bound \\ \hline\hline
  19 & 0.2366 \\
  29 & 0.2115 \\
  39 & 0.1970 \\
  49 & 0.1872 \\
  59 & 0.1799 \\
  69 & 0.1742 \\
  79 & 0.1696 \\
  89 & 0.1658 \\
  99 & 0.1625
 \end{tabular}
 \label{tab:spa-results}
\end{table}


\subsubsection{Atom-cavity model:} \label{ex:duan-kimble}

Consider a three-level atom coupled to an optical cavity, which itself is coupled to a single external coherent field ($m=1$) \cite{DK04}.
Let ${\cal H}_0 = \mathbb{C}^3 \otimes \ell^2$. We will use $\ket{e} = (1,0,0)^\top$, $\ket{+} = (0,1,0)^\top$, and $\ket{-} = (0,0,1)^\top$ to denote the canonical basis vectors in $\mathbb{C}^m$. We also consider $\ket{n}$, denoting the normalized $n$-photon Fock state of the system, as the basis vector of $\ell^2$ (as in the previous example). Let ${\cal D}_0 = \mathbb{C}^3 \otimes {\rm span}\{ \ket{n} \mid n \in \mathbb{Z}_+ \}$.
This atom-cavity system is then described by the following parameters:
\begin{align*}
 S =I\otimes I, \qquad
 L = I \otimes \sqrt{\lambda} a, \qquad
 H = \imath \chi\left( \sigma_{+}\otimes a - \sigma_{-}\otimes a^*  \right),
\end{align*}
where $\lambda, \chi > 0$, $\sigma_{+} = \ket{e}\bra{+}$, and $\sigma_{-} = \ket{+}\bra{e}$.
From Lemma 12 of \cite{BvHS08}, we have that Condition \ref{cond:cocycle}\ref{cond:cocycle-1} and \ref{cond:core}\ref{cond:core-1} holds for this atom-cavity model.

Consider ${\cal H}^{(k)} = \mathbb{C}^3 \otimes {\rm span}\left\{ \ket{n} \mid n = 0,1,2,\ldots,k \right\}$ and a system approximation of the form \eqref{eq:approx-system}. That is, we are approximating the dynamics of the harmonic oscillator. 
Recall that Conditions \ref{cond:cocycle}\ref{cond:cocycle-2} and \ref{cond:core}\ref{cond:core-2} hold immediately because ${\cal D}^{(k)} = {\cal H}^{(k)}$ has finite dimension.
We also see that ${\cal H}^{(k)} \subset {\cal D}_0$. Hence, Assumption \ref{assump:domains} holds. Similar to the previous example, we note that ${\cal M}^{(k)} = \mathbb{C}^3\otimes\ket{k+1}$. 
Now consider ${\cal K}^{(k)} = \mathbb{C}^3\otimes\ket{k}$ and ${\gamma^{(\alpha\beta)}_k} = \frac{1}{2}\left(\lambda k + |\alpha-\beta|^2\right)$.
Using similar derivation to the previous example, we will now show that Assumptions \ref{assump:approx} and \ref{assump:convergence} hold for the atom-cavity model and the approximation.

\paragraph*{Assumption \ref{assump:approx}\ref{assump:subspaceK1}:}
Note that
\begin{align*}
 \left.\left({\cal L}^{(\alpha\beta)} - {\cal L}^{(k;\alpha\beta)}\right)\right|_{{\cal K}^{(k)}} \hspace{-0.5em}
 &= \left(\beta\sqrt{\lambda} + \chi\sigma_{-}\right) \left. a^* \right|_{{\cal K}^{(k)}}.
\end{align*}
Thus, we have that Assumption \ref{assump:approx}\ref{assump:subspaceK1} holds with $q^{(k;\alpha\beta)}_{\cal L} = \sqrt{k+1}\left(|\beta|\sqrt{\lambda} + \chi\right)$.

\paragraph*{Assumption \ref{assump:approx}\ref{assump:subspaceK2}:}
This assumption follows for the defined ${\cal K}^{(k)}$ because \\ $\ker\left(\left.\left({\cal L}^{(\alpha\beta)} - {\cal L}^{(k;\alpha\beta)}\right)\right|_{{\cal H}^{(k)}} \right) = \mathbb{C}^3\otimes {\rm span}\left\{ \ket{n} \mid n = 0,1,2,\ldots,k-1 \right\}$ for all $\alpha,\beta \in \mathbb{C}$.

\paragraph*{Assumption \ref{assump:approx}\ref{assump:dissipate-approx}:}
Note that ${L^{(k)}}^* P_{{\cal K}^{(k)}} = 0$,
$P_{{\cal K}^{(k)}} L^{(k)} P_{{\cal H}^{(k)}} = 0$, \\
$P_{{\cal K}^{(k)}} {L^{(k)}}^* L^{(k)} P_{{\cal H}^{(k)}\ominus{\cal K}^{(k)}} = 0$, and 
$P_{{\cal K}^{(k)}} H P_{{\cal H}^{(k)}\ominus{\cal K}^{(k)}} =  -\imath\chi P_{{\cal K}^{(k)}} \sigma_{-} a^* P_{{\cal H}^{(k)}\ominus{\cal K}^{(k)}}$.
Also, for any $u \in {\cal H}^{(k)}$,
$P_{{\cal K}^{(k)}}{L^{(k)}}^* L^{(k)} P_{{\cal K}^{(k)}}u =  \lambda k P_{{\cal K}^{(k)}}u$.
From these identities and the fact that $H^{(k)}$ is self-adjoint, we have for any $u \in {\cal H}^{(k)}$ that
\begin{align*}
 \Re\{\langle P_{{\cal K}^{(k)}} {\cal L}^{(k;\alpha\beta)} u, P_{{\cal K}^{(k)}} u\rangle \} 
 &= \Re\{\langle P_{{\cal K}^{(k)}} {\cal L}^{(k;\alpha\beta)} \left(P_{{\cal K}^{(k)}} + P_{{\cal H}^{(k)}\ominus{\cal K}^{(k)}}\right) u, P_{{\cal K}^{(k)}} u\rangle \} \nonumber \\
 &= -\frac{1}{2} \left(\lambda k + |\alpha|^2 + |\beta|^2 - \Re\{\alpha^*\beta\} \right) \left\| P_{{\cal K}^{(k)}} u \right\|^2 \nonumber \\
 &\quad + \Re\left\{\left< P_{{\cal K}^{(k)}} \left(\beta\sqrt{\lambda} + \chi\sigma_{-}\right)   a^* P_{{\cal H}^{(k)}\ominus{\cal K}^{(k)}} u, P_{{\cal K}^{(k)}} u\right> \right\} \nonumber \\
 &= -g(k,\alpha,\beta) \left\| P_{{\cal K}^{(k)}} u \right\|^2 + h(k,\alpha,\beta,u).
\end{align*}
Noticing that $|\alpha|^2 + |\beta|^2 - 2\Re\{\alpha^*\beta\} = \braket{\alpha,\alpha} + \braket{\beta,\beta} - \braket{\alpha,\beta} - \overline{\braket{\alpha,\beta}} = |\alpha-\beta|^2$.
Also, note that 
\begin{align*}
 &\left|\Re\left\{\left< P_{{\cal K}^{(k)}} \left(\beta\sqrt{\lambda} + \chi\sigma_{-}\right)   a^* P_{{\cal H}^{(k)}\ominus{\cal K}^{(k)}} u, P_{{\cal K}^{(k)}} u\right> \right\}\right| \nonumber \\
 &\leq \left\|P_{{\cal K}^{(k)}} \left(\beta\sqrt{\lambda} + \chi\sigma_{-}\right)   a^* P_{{\cal H}^{(k)}\ominus{\cal K}^{(k)}} u\right\|  \left\| P_{{\cal K}^{(k)}} u \right\| \nonumber \\
 &\leq \sqrt{k}\left(\chi + |\beta|\sqrt{\lambda}\right) \|u\| \left\| P_{{\cal K}^{(k)}} u \right\|.
\end{align*}
Therefore, we have that Assumption \ref{assump:approx}\ref{assump:dissipate-approx} holds for the defined ${\cal K}^{(k)}$ and the defined ${\gamma^{(\alpha\beta)}_k}$ with $q^{(k;\alpha\beta)}_{a} = \sqrt{k}\left(\chi + |\beta|\sqrt{\lambda}\right)$.

\paragraph*{Assumption \ref{assump:approx}\ref{assump:dissipate-true}:}
For any $u \in {\cal H}_0$, note that ${L}^* L P_{{\cal M}^{(k)}}u = \lambda(k+1) P_{{\cal M}^{(k)}}u$ and 
$\imath H P_{{\cal M}^{(k)}}u = \imath \left(\Delta (k+1) + \chi k(k+1) \right) P_{{\cal M}^{(k)}}u$. Also, 
Thus, we have that
\begin{align*}
 &\Re\{\langle T^{(\alpha\beta)}_t {\cal L}^{(\alpha\beta)} P_{{\cal M}^{(k)}} u, T^{(\alpha\beta)}_t P_{{\cal M}^{(k)}} u\rangle \} \nonumber \\
 &=-\frac{1}{2} \bigg(\lambda(k+1) + |\alpha|^2 + |\beta|^2 -2\Re\{\alpha^*\beta\} \bigg)  \left\| T^{(\alpha\beta)}_t P_{{\cal M}^{(k)}} u \right\|^2 \nonumber \\
 &\quad + \Re\left\{\left< T^{(\alpha\beta)}_t \sqrt{\lambda}\left(\beta a^* - \alpha^* a \right) P_{{\cal M}^{(k)}} u, T^{(\alpha\beta)}_t P_{{\cal M}^{(k)}} u\right> \right\} \nonumber \\ 
 &\quad + \Re\left\{\left< T^{(\alpha\beta)}_t \chi\left( \sigma_{-} a^* - \sigma_{+} a \right) P_{{\cal M}^{(k)}} u, T^{(\alpha\beta)}_t P_{{\cal M}^{(k)}} u\right> \right\} \nonumber \\ 
 &= -\hat{g}(k,\alpha,\beta) \left\| T^{(\alpha\beta)}_t P_{{\cal M}^{(k)}} u \right\|^2 + \hat{h}(t,k,\alpha,\beta,u).
\end{align*}
Similar to the previous derivation, using that $T^{(\alpha\beta)}_t$ is a contraction, we see that Assumption \ref{assump:approx}\ref{assump:dissipate-true} holds for the defined ${\cal K}^{(k)}$ and ${\gamma^{(\alpha\beta)}_k}$ with
$q^{(k;\alpha\beta)}_{e} = \sqrt{k+1} \left( \chi + |\alpha| \sqrt{\lambda} \right) + \sqrt{k+2} \left( \chi + |\beta| \sqrt{\lambda} \right)$.

\paragraph*{Assumption \ref{assump:convergence}:}
From the defined ${\gamma^{(\alpha\beta)}_k}$, $q^{(k;\alpha\beta)}_{\cal L}$, $q^{(k;\alpha\beta)}_{a}$, $q^{(k;\alpha\beta)}_{e}$, we see that this assumption holds for any $r,s \in \mathbb{Z}_+$ such that $r+s \geq 3$. 

Finally, because Assumptions \ref{assump:domains}-\ref{assump:convergence} hold, Lemma \ref{lemma:main1}, Corollary \ref{cor:main2}, and Theorem \ref{thm:main3} can be applied to obtain error bounds on the finite-dimensional approximations. 


\section{Error bounds for adiabatic elimination approximations} \label{sec:main-adiabatic}

When an open quantum system comprises subsystems evolving at two well-separated timescales, the system dynamics can be approximated by eliminating the fast variables from the model description. This method is known as adiabatic elimination in the physics literature and  singular perturbation in the applied mathematics literature. In this section, we will establish error bounds for this type of finite dimensional approximation of open quantum systems, when the slow subsystem lives on a finite-dimensional subspace.

Let $U^{(k)}_t$ satisfying \eqref{eq:QSDE-k} describe  the time evolution of the original Markov open quantum system to be approximated.
We set ${\cal H}^{(k)} = {\cal H}$ and ${\cal D}^{(k)} = {\cal D}$. We also let $U_t$ satisfying \eqref{eq:QSDE} describe the time evolution of the adiabatic elimination approximation defined on a finite-dimensional subspace ${\cal H}_0 \subset {\cal H}$. In this setting, ${\cal D}_0 = {\cal H}_0$.
We note that Conditions \ref{cond:cocycle}\ref{cond:cocycle-1} and \ref{cond:core}\ref{cond:core-1} hold immediately because ${\cal H}_0$ is finite dimensional. As in \cite{BvHS08}, we assume the following.

\begin{assumption}[Singular scaling] \label{assump:AE-scaling}
For $i,j = 1,2,\ldots,m$, there exists operators $Y$, $Y^*$, $A$, $A^*$, $B$, $B^*$, $F_j$, $F_j^*$, $G_j$, $G_j^*$, $W_{ij}$, $ W_{ij}^*$ with the common invariant domain ${\cal D}$ such that
\begin{align*}
 \left[\imath H^{(k)} - \frac{1}{2}\sum_{i=1}^{m}{L^{(k)}_i}^*  L^{(k)}_i\right] = k^2 Y + kA + B, \qquad
 {L^{(k)}_j}^* = kF_j + G_j, \qquad
 {S^{(k)}_{ji}}^* = W_{ij}.
\end{align*}
\end{assumption}
\begin{assumption}[Structural requirements] \label{assump:AE-structure}
The subspace ${\cal H}_0 \subset {\cal H}$ is a closed subspace such that 
\begin{enumerate} \renewcommand{\theenumi}{\alph{enumi}}
 \item ${\cal D}_0 = P_{{\cal H}_0} {\cal D} \subset {\cal D}$
 \item $Y P_{{\cal H}_0} = 0$ on ${\cal D}$
 \item There exists $\tilde{Y}$, $\tilde{Y}^*$ with the common invariant domain ${\cal D}$ so that $\tilde{Y}Y = Y\tilde{Y} = P_{{\cal H}_0^\perp}$.
 \item $F_j^* P_{{\cal H}_0} = 0$ on ${\cal D}$ for all $j = 1,2,\ldots,m$
 \item $P_{{\cal H}_0} A P_{{\cal H}_0} = 0$ on ${\cal D}$.
\end{enumerate} 
\end{assumption}
\begin{assumption}[Limit coefficients] \label{assump:AE-coefficients}
The approximating system operators $(S,L,H)$ are such that
\begin{align*}
 S_{ji}^* &= \sum_{\ell=1}^m P_{{\cal H}_0} W_{i\ell} \left( F_\ell^* \tilde{Y} F_j + \delta_{\ell j} \right) P_{{\cal H}_0}, \\
 L_j^* &= P_{{\cal H}_0} \left( G_j - A\tilde{Y}F_j \right)P_{{\cal H}_0}, \\
 H &= \Im\left\{ P_{{\cal H}_0} \left( B - A\tilde{Y}A \right)P_{{\cal H}_0} \right\}.
\end{align*}
\end{assumption}

Let us now re-state an important result.
\begin{lemma}[{\cite[Lemma 10]{BvHS08}}]
Suppose Assumptions \ref{assump:AE-scaling}, \ref{assump:AE-structure}, and \ref{assump:AE-coefficients} hold.
The linear operators $S,L,H$, defined in Assumption \ref{assump:AE-coefficients}, have the common invariant domain ${\cal D}_0$.
The operator $S$ is unitary and $H$ is self-adjoint.
\end{lemma}

For any $\alpha,\beta \in \mathbb{C}^m$ and any $v \in {\cal D}$, let us define
\begin{align*}
 A^{(\alpha\beta)}v &:= \left(A + \sum_{j=1}^m F_j \beta_j - \sum_{i,j=1}^m \alpha_i^* W_{ij} F_j^*\right) v, \\
 B^{(\alpha\beta)}v &:= \left(-\frac{|\alpha|^2+|\beta|^2}{2} + B + \sum_{j=1}^m G_j\beta_j + \sum_{i,j=1}^m \alpha_i^* W_{ij} \left(\beta_j - G_j^*\right)\right) v.
\end{align*}
We now introduce an additional assumption required in obtaining our error bound results.

\begin{assumption}[Boundedness of operators] \label{assump:AE-boundedness}
For any $k \in \mathbb{Z}_+$ and any $\alpha,\beta \in \mathbb{C}^m$, 
we have that
\begin{align*}
 M^{(k;\alpha,\beta)}_1 &:= \left\| \left. \tilde{Y}P_{{\cal H}_0^\perp} \left(A^{(\alpha\beta)} - \frac{1}{k}\left(B^{(\alpha\beta)} - A^{(\alpha\beta)}\tilde{Y}A^{(\alpha\beta)} \right) \right) \right|_{{\cal H}_0} \right\|, \\
 M^{(k;\alpha,\beta)}_2 &:= \left\| \tilde{Y}P_{{\cal H}_0^\perp} \left(A^{(\alpha\beta)} - \frac{1}{k}\left(B^{(\alpha\beta)} - A^{(\alpha\beta)}\tilde{Y}A^{(\alpha\beta)} \right) \right){\cal L}^{(\alpha\beta)} \right. \\
 &\quad \left. + \left. \left(B^{(\alpha\beta)}\tilde{Y}A^{(\alpha\beta)} + \left( A^{(\alpha\beta)} - \frac{1}{k}B^{(\alpha\beta)}\right)\tilde{Y} P_{{\cal H}_0^\perp} \left(B^{(\alpha\beta)} - A^{(\alpha\beta)}\tilde{Y}A^{(\alpha\beta)} \right)\right)\right|_{{\cal H}_0} \right\|
\end{align*} 
are finite (i.e., $0 \leq M^{(k;\alpha\beta)}_1, M^{(k;\alpha\beta)}_2 < \infty$).
\end{assumption}

\begin{lemma} \label{lemma:AE-semigroup}
Suppose Assumptions \ref{assump:AE-scaling}, \ref{assump:AE-structure}, and \ref{assump:AE-coefficients} hold.
Then for any $k \in \mathbb{Z}_+$, any $\alpha,\beta \in \mathbb{C}^m$, any $u \in {\cal H}_0$, and any $t \geq 0$, it holds that
\begin{align}
 \left\| \left( T^{(k;\alpha\beta)}_t - T^{(\alpha\beta)}_t \right) u \right\|
 &\leq \frac{1}{k} \left( 2M^{(k;\alpha,\beta)}_1 + t M^{(k;\alpha,\beta)}_2 \right) \|u\|.
 \label{eq:AE-bound1}
 \end{align}
Moreover, if $T_t^{(\alpha \beta)}$ and $T_t^{(k;\alpha \beta)}$ are also norm continuous for each $\alpha, \beta$, with 
 $\| (I-T_t^{(\alpha \beta)})u\| \leq N_1^{(\alpha \beta)}(t) \|u\|$ and $\| (I-T_t^{(k;\alpha \beta)})u\| \leq N_2^{(k;\alpha \beta)}(t)\|u\|$ for some continuous nonnegative functions $N_1^{(\alpha \beta)}$ and $N_2^{(k;\alpha \beta)}$, then 
 \begin{align}
 \left\| \left( T^{(k;\alpha\beta)}_t - T^{(\alpha\beta)}_t \right) u \right\|
 &\leq \frac{1}{k} \left( M^{(k;\alpha,\beta)}_1(N_1^{(\alpha \beta)}(t) + N_2^{(k;\alpha \beta)}(t)) + t M^{(k;\alpha,\beta)}_2 \right) \|u\|.
 \label{eq:AE-bound2}
 \end{align}
\end{lemma}
\begin{proof}
First note, from the definition of a strongly continuous semigroup, that \cite{CZ95}
\begin{enumerate} \renewcommand{\theenumi}{\roman{enumi}}
 \item  $T_0^{(\alpha\beta)} = T_0^{(k;\alpha\beta)} = I$ for all $k \in \mathbb{Z}_+$,
 \item $\frac{d}{dt} T_t^{(\alpha\beta)} u = {\cal L}^{(\alpha\beta)} u$  for all $u \in {\cal H}_0$ (since ${\cal H}_0$ is finite-dimensional),
 \item $\frac{d}{dt} T_t^{(k;\alpha\beta)} u = {\cal L}^{(k;\alpha\beta)} u$ for all $u \in {\rm Dom}({\cal L}^{(k;\alpha\beta)})$.
\end{enumerate}
From the above properties and Assumption \ref{assump:AE-structure}, we can write for all $u \in {\cal H}_0$ and all $t \geq 0$ that
\begin{align}
 \frac{d}{dt} \left(T^{(k;\alpha\beta)}_{t} - T^{(\alpha\beta)}_{t}\right) u 
 &= \left( {\cal L}^{(k;\alpha\beta)}T^{(k;\alpha\beta)}_{t} - {\cal L}^{(\alpha\beta)}T^{(\alpha\beta)}_{t}\right) u \nonumber \\
 &= {\cal L}^{(k;\alpha\beta)}\left(T^{(k;\alpha\beta)}_{t} - T^{(\alpha\beta)}_{t}\right) u + \left. \left( {\cal L}^{(k;\alpha\beta)} - {\cal L}^{(\alpha\beta)}\right) \right|_{\mathcal{H}_0} T^{(\alpha\beta)}_{t} u
 \label{eq:cauchy-problem2}
\end{align}
with $\left(T^{(\alpha\beta)}_0 - T^{(k;\alpha\beta)}_0\right) u = 0$. 
Since ${\cal L}^{(k;\alpha\beta)}$ is a generator of a semigroup on a Hilbert space, $\left(T^{(k;\alpha\beta)}_0 - T^{(\alpha\beta)}_0\right) u \in {\rm Dom}({\cal L}^{(k;\alpha\beta)})$ (due to Assumption \ref{assump:AE-structure}), 
and $\left. \left( {\cal L}^{(k;\alpha\beta)} - {\cal L}^{(\alpha\beta)}\right) \right|_{\mathcal{H}_0} T^{(\alpha\beta)}_{t} u \in {\cal C}^{1}([0,t];{\cal H})$ (the class of continuously differentiable functions from $[0,t]$ to ${\cal H}$), a unique solution of \eqref{eq:cauchy-problem2} exists and is given by \cite[Thm 3.1.3]{CZ95}
\begin{align}
 \left(T^{(k;\alpha\beta)}_{t} - T^{(\alpha\beta)}_{t}\right) u = \int_0^{t} T^{(k;\alpha\beta)}_{{t}-\tau} \left( {\cal L}^{(k;\alpha\beta)} - {\cal L}^{(\alpha\beta)}\right) T^{(\alpha\beta)}_{\tau} u d\tau
 \label{eq:semigroup1}
\end{align}
for all $t \geq 0$ and all $u \in {\cal H}_0$.

From Assumption \ref{assump:AE-scaling}, we now note that 
\begin{align*}
 {\cal L}^{(k;\alpha\beta)} u = \left( k^2 Y + k A^{(\alpha\beta)} + k B^{(\alpha\beta)} \right) u
\end{align*}
for any $u \in {\cal D}$ and any $\alpha,\beta \in \mathbb{C}^m$.
Using Assumptions \ref{assump:AE-scaling} and \ref{assump:AE-coefficients}, it has been shown that \cite[p. 3146]{BvHS08}
\begin{align*}
 {\cal L}^{(\alpha\beta)} u = P_{{\cal H}_0} \left( B^{(\alpha\beta)} - A^{(\alpha\beta)}\tilde{Y}A^{(\alpha\beta)} \right) u
\end{align*}
for any $u \in {\cal D}_0$ and any $\alpha,\beta \in \mathbb{C}^m$.

From the above identities and Assumption \ref{assump:AE-scaling}, for any $v \in {\cal H}_0$, we have that \cite{BvHS08}
\begin{align*}
 &\left( {\cal L}^{(k;\alpha\beta)} - {\cal L}^{(\alpha\beta)} \right) v \\
 &= \left( {\cal L}^{(k;\alpha\beta)} - {\cal L}^{(\alpha\beta)} - \frac{1}{k}{\cal L}^{(k;\alpha\beta)}\tilde{Y}A^{(\alpha\beta)} + \frac{1}{k}{\cal L}^{(k;\alpha\beta)}\tilde{Y}A^{(\alpha\beta)} \right. \\
 &\quad - \frac{1}{k^2}{\cal L}^{(k;\alpha\beta)}\tilde{Y}P_{{\cal H}_0^\perp} \left(B^{(\alpha\beta)} - A^{(\alpha\beta)}\tilde{Y}A^{(\alpha\beta)} \right) \\
 &\quad \left. + \frac{1}{k^2}{\cal L}^{(k;\alpha\beta)}\tilde{Y}P_{{\cal H}_0^\perp} \left(B^{(\alpha\beta)} - A^{(\alpha\beta)}\tilde{Y}A^{(\alpha\beta)} \right) \right)v \\
 &= \left[ \frac{1}{k}{\cal L}^{(k;\alpha\beta)}\tilde{Y}A^{(\alpha\beta)} - \frac{1}{k} \left(B^{(\alpha\beta)}\tilde{Y}A^{(\alpha\beta)} + A^{(\alpha\beta)} \tilde{Y} P_{{\cal H}_0^\perp} \left( B^{(\alpha\beta)} - A^{(\alpha\beta)}\tilde{Y}A^{(\alpha\beta)} \right) \right)  \right.  \\
 &\quad + \frac{1}{k^2}{\cal L}^{(k;\alpha\beta)}\tilde{Y}P_{{\cal H}_0^\perp} \left(B^{(\alpha\beta)} - A^{(\alpha\beta)}\tilde{Y}A^{(\alpha\beta)} \right) \\
 &\quad \left. - \frac{1}{k^2} B^{(\alpha\beta)} \tilde{Y} P_{{\cal H}_0^\perp} \left( B^{(\alpha\beta)} - A^{(\alpha\beta)}\tilde{Y}A^{(\alpha\beta)} \right)   \right] v.
\end{align*}
We also note, using integration by parts (in a similar manner to \cite[p. 898]{LM88} and \cite[Eq. (2.2)]{IK98}), that
\begin{align*}
 \int_0^t T^{(k;\alpha\beta)}_{t-\tau} {\cal L}^{(k;\alpha\beta)} \tilde{Y}A^{(\alpha\beta)} T^{(\alpha\beta)}_\tau u d\tau
 &= \left( T^{(k;\alpha\beta)}_{t}\tilde{Y}A^{(\alpha\beta)} - \tilde{Y}A^{(\alpha\beta)}T^{(\alpha\beta)}_t \right)u \\
 &\quad + \int_0^t T^{(k;\alpha\beta)}_{t-\tau} \tilde{Y}A^{(\alpha\beta)} {\cal L}^{(\alpha\beta)} T^{(\alpha\beta)}_\tau u d\tau
\end{align*}
where we have used that $- \frac{d}{d\tau} T^{(k;\alpha\beta)}_{t-\tau} u = T^{(k;\alpha\beta)}_{t-\tau} {\cal L}^{(k;\alpha\beta)} u$ 
and $\frac{d}{d\tau} \tilde{Y}A^{(\alpha\beta)} T^{(\alpha\beta)}_{\tau} u = \tilde{Y}A^{(\alpha\beta)} {\cal L}^{(\alpha\beta)} T^{(\alpha\beta)}_{\tau} u$ \cite{CZ95}. Similarly, we have that
\begin{align*}
 &\int_0^t T^{(k;\alpha\beta)}_{t-\tau} {\cal L}^{(k;\alpha\beta)} \tilde{Y}P_{{\cal H}_0^\perp} \left( B^{(\alpha\beta)} - A^{(\alpha\beta)}\tilde{Y}A^{(\alpha\beta)} \right)  T^{(\alpha\beta)}_\tau u d\tau \\
 &= \left( T^{(k;\alpha\beta)}_{t}\tilde{Y}P_{{\cal H}_0^\perp} \left( B^{(\alpha\beta)} - A^{(\alpha\beta)}\tilde{Y}A^{(\alpha\beta)} \right) - \tilde{Y}P_{{\cal H}_0^\perp} \left( B^{(\alpha\beta)} - A^{(\alpha\beta)}\tilde{Y}A^{(\alpha\beta)} \right)T^{(\alpha\beta)}_t \right)u \\
 &\quad + \int_0^t T^{(k;\alpha\beta)}_{t-\tau} \tilde{Y}P_{{\cal H}_0^\perp} \left( B^{(\alpha\beta)} - A^{(\alpha\beta)}\tilde{Y}A^{(\alpha\beta)} \right) {\cal L}^{(\alpha\beta)} T^{(\alpha\beta)}_\tau u d\tau.
\end{align*}

Now, because $T^{(\alpha\beta)}_t u \in {\cal H}_0$ for all $t \geq 0$, by substituting the above three identities into \eqref{eq:semigroup1}, we have that
\begin{align*}
 &\left(T^{(k;\alpha\beta)}_{t} - T^{(\alpha\beta)}_{t}\right) u \\
 &= \frac{1}{k} \left[ T^{(k;\alpha\beta)}_{t} \tilde{Y}P_{{\cal H}_0^\perp} \left(A^{(\alpha\beta)} - \frac{1}{k}\left(B^{(\alpha\beta)} - A^{(\alpha\beta)}\tilde{Y}A^{(\alpha\beta)} \right) \right) u \right. \\
 &\quad - \tilde{Y}P_{{\cal H}_0^\perp} \left(A^{(\alpha\beta)} - \frac{1}{k}\left(B^{(\alpha\beta)} - A^{(\alpha\beta)}\tilde{Y}A^{(\alpha\beta)} \right) \right) T^{(\alpha\beta)}_t u  \\
 &\quad + \int_0^t T^{(k;\alpha\beta)}_{t-\tau} \left( \tilde{Y}P_{{\cal H}_0^\perp} \left(A^{(\alpha\beta)} - \frac{1}{k}\left(B^{\alpha\beta} - A^{(\alpha\beta)}\tilde{Y}A^{(\alpha\beta)} \right) \right) \right){\cal L}^{(\alpha\beta)} T^{(\alpha\beta)}_\tau u d\tau \\
 &\quad + \int_0^t T^{(k;\alpha\beta)}_{t-\tau} \left(B^{(\alpha\beta)}\tilde{Y}A^{(\alpha\beta)} + \left( A^{(\alpha\beta)} - \frac{1}{k}B^{(\alpha\beta)}\right)\tilde{Y} P_{{\cal H}_0^\perp} \left(B^{(\alpha\beta)} - A^{(\alpha\beta)}\tilde{Y}A^{(\alpha\beta)} \right)\right)  \\
 &\qquad \quad \left. \times  T^{(\alpha\beta)}_\tau u d\tau \right].
\end{align*}
Here, we have used Assumption \ref{assump:AE-structure} that $P_{{\cal H}_0} AP_{{\cal H}_0} = 0$.
We then have that
\begin{align*}
 \left\|  T^{(k;\alpha\beta)}_{t} - T^{(\alpha\beta)}_{t} \right\|
 &\leq \frac{1}{k} \left[ M^{(k;\alpha\beta)}_1 \left(\left\|T^{(k;\alpha\beta)}_t\right\| +  \left\|T^{(\alpha\beta)}_t\right\|\right) \right. \nonumber \\ 
 &\quad \left. + M^{(k;\alpha\beta)}_2 \int_0^t \left\|T^{(k;\alpha\beta)}_t\right\| \left\|T^{(\alpha\beta)}_t\right\| d\tau \right] \|u\|.
\end{align*}
The result \eqref{eq:AE-bound1} then follows from the fact that $T^{(k;\alpha\beta)}_t$ and $T^{(k;\alpha\beta)}_t$ are contraction semigroups. Moreover, when $T_t^{(\alpha \beta)}$ and $T_t^{(k;\alpha \beta)}$ are also norm continuous and satisfy $\| (I- T_t^{(\alpha \beta)}) u\|  \leq N_1^{(\alpha \beta)}(t) \|u\|$ and   $\| (I- T_t^{(k;\alpha \beta)}) u\|  \leq N_2^{(k;\alpha \beta)}(t)\|u\|$ as stipulated in the lemma, we have 
\begin{align*}
 &\left(T^{(k;\alpha\beta)}_{t} - T^{(\alpha\beta)}_{t}\right) u \\
&= \frac{1}{k} \left[ T^{(k;\alpha\beta)}_{t} \tilde{Y}P_{{\cal H}_0^\perp} \left(A^{(\alpha\beta)} - \frac{1}{k}\left(B^{(\alpha\beta)} - A^{(\alpha\beta)}\tilde{Y}A^{(\alpha\beta)} \right) \right) (I-T_t^{(\alpha \beta)})u \right. \\
 &\quad - (I-T_t^{(k;\alpha \beta)})\tilde{Y}P_{{\cal H}_0^\perp} \left(A^{(\alpha\beta)} - \frac{1}{k}\left(B^{(\alpha\beta)} - A^{(\alpha\beta)}\tilde{Y}A^{(\alpha\beta)} \right) \right) T^{(\alpha\beta)}_t u  \\
 &\quad + \int_0^t T^{(k;\alpha\beta)}_{t-\tau} \left( \tilde{Y}P_{{\cal H}_0^\perp} \left(A^{(\alpha\beta)} - \frac{1}{k}\left(B^{\alpha\beta} - A^{(\alpha\beta)}\tilde{Y}A^{(\alpha\beta)} \right) \right) \right){\cal L}^{(\alpha\beta)} T^{(\alpha\beta)}_\tau u d\tau \\
 &\quad + \int_0^t T^{(k;\alpha\beta)}_{t-\tau} \left(B^{(\alpha\beta)}\tilde{Y}A^{(\alpha\beta)} + \left( A^{(\alpha\beta)} - \frac{1}{k}B^{(\alpha\beta)}\right)\tilde{Y} P_{{\cal H}_0^\perp} \left(B^{(\alpha\beta)} - A^{(\alpha\beta)}\tilde{Y}A^{(\alpha\beta)} \right)\right)  \\
 &\qquad \quad \left. \times  T^{(\alpha\beta)}_\tau u d\tau \right].
\end{align*}
From this and the contractivity of $T_t^{(\alpha \beta)}$ and $T_t^{(k;\alpha \beta)}$, it follows that
\begin{align*}
\left\|  T^{(k;\alpha\beta)}_{t} - T^{(\alpha\beta)}_{t} \right\|
&\leq \frac{1}{k} \left[ M^{(k;\alpha\beta)}_1 \left(N_1^{(\alpha \beta)}(t) + N_2^{(k;\alpha \beta)}(t) \right) + M^{(k;\alpha\beta)}_2 t \right] \|u\|.
\end{align*}
This establishes the lemma statement.
\end{proof}

Recall that $\mathfrak{S}' \subset L^2([0,T];\mathbb{C}^m)$ is the set of all simple functions in $L^2([0,T];\mathbb{C}^m)$, which is dense in $L^2([0,T];\mathbb{C}^m)$.
Similar to the previous section, let ${\cal U}_0 = \{u\otimes e(f) \mid u \in {\cal H}_0, f \in \mathfrak{S}' \}$ and ${\cal U}^{(k)} = \{ u\otimes e(f) \mid u \in {\cal H}^{(k)}, f \in \mathfrak{S}' \}$ for $k>0$.

\begin{lemma} \label{lemma:AE-main1}
Suppose Assumptions \ref{assump:AE-scaling}, \ref{assump:AE-structure}, \ref{assump:AE-coefficients}, and \ref{assump:AE-boundedness} hold. Then for any $k \in \mathbb{Z}_+$, any $\psi_1,\psi_2 \in {\cal U}_0$, and any $t \in [0,T]$ with $0 \leq T < \infty$, we have that
\begin{align}
 \left| \left< \left( U^{(k)}_t - U_t \right)^* \psi_1,\psi_2 \right>\right|
 \leq \frac{1}{k} \sum_{i=0}^\ell \left( 2M^{(k;\alpha_1(i)\alpha_2(i))}_1 + (t_{i+1}-t_i)M^{(k;\alpha_1(i)\alpha_2(i))}_2 \right) \| \psi_1\| \| \psi_2\|
 \label{eq:AE-Ut-err1}
\end{align}
where $t_0 = 0$, $t_{\ell+1} = t$, and $ f_j = \sum_{i=0}^{\ell} \alpha_j(i)\mathds{1}_{[t_i,t_{i+1})}$ for $j = 1,2$. Moreover, if $T_t^{(\alpha \beta)}$ and $T_t^{(k;\alpha \beta)}$ are also norm continuous for each $\alpha, \beta$, with $\| (I-T_t^{(\alpha \beta)})u\| \leq N_1^{(\alpha \beta)}(t) \|u\|$ and $\| (I-T_t^{(k;\alpha \beta)})u\| \leq N_2^{(k;\alpha \beta)}(t)\|u\|$ for some continuous nonnegative functions $N_1^{(\alpha \beta)}$ and $N_2^{(k;\alpha \beta)}$, then 
\begin{align}
 \left| \left< \left( U^{(k)}_t - U_t \right)^* \psi_1,\psi_2 \right>\right|
 &\leq \frac{1}{k} \sum_{i=0}^\ell \left( M^{(k;\alpha_1(i)\alpha_2(i))}_1 \bigl(N_1^{(\alpha \beta)}(t_{i+1}-t_i) + N_2^{(k;\alpha \beta)}(t_{i+1}-t_i)\bigr) \right. \nonumber \\ 
 &\quad \left. + (t_{i+1}-t_i)M^{(k;\alpha_1(i)\alpha_2(i))}_2 \right) \| \psi_1\| \| \psi_2\|
 \label{eq:AE-Ut-err1b}
\end{align}

\end{lemma}
\begin{proof}
This proof follows similar arguments to the proof of Lemma \ref{lemma:main1} using the result \eqref{eq:AE-bound1} established in Lemma \ref{lemma:AE-semigroup}.
\end{proof}

\begin{corollary}\label{cor:AE-main2}
Suppose Assumptions \ref{assump:AE-scaling}, \ref{assump:AE-structure}, \ref{assump:AE-coefficients}, and \ref{assump:AE-boundedness} hold.
For any $t \in [0,T]$ with $0 < T < \infty$, any $\psi_1 = u_1\otimes e(f_1),\psi_2 = u_2\otimes e(f_2) \in {\cal H}_0\otimes{\cal F}$, we have that
\begin{align}
 \left| \left< \left(U^{(k)}_t - U_t\right)^* \psi_1,\psi_2 \right> \right| 
 &\leq 2 \Big( \|u_1\| \|e(f_1)-e(f'_1)\| \|\psi_2\| + \|u_2\| \|e(f_2)-e(f'_2)\| \|\psi_1\| \Big) \nonumber \\
 &\quad + \frac{1}{k} \sum_{i=0}^\ell \left( 2M^{(k;\alpha_1(i)\alpha_2(i))}_1 + (t_{i+1}-t_i)M^{(k;\alpha_1(i)\alpha_2(i))}_2 \right) \|\psi_1'\|  \|\psi_2'\|.
 \label{eq:AE-Ut-err2}
\end{align}
for any $\psi_j' = u_j  \otimes e(f_j') \in {\cal U}^{(k)}$ with for some  $\ell \in \mathbb{Z}_+$ and a sequence $t_0=0 <t_1 < \ldots <t_{\ell}<t_{\ell+1}=t$ such that $f'_1 =\sum_{i=0}^{\ell} \alpha_1(i)\mathds{1}_{[t_i,t_{i+1})}$ and $f'_2 = \sum_{i=0}^{\ell}\alpha_2(i) \mathds{1}_{[t_i,t_{i+1})}$.  If in addition,  $T_t^{(\alpha \beta)}$ and $T_t^{(k;\alpha \beta)}$ are also norm continuous for each $\alpha,\beta$, with 
 $\| (I-T_t^{(\alpha \beta)})u\| \leq N_1^{(\alpha \beta)}(t) \|u\|$ and $\| (I-T_t^{(k;\alpha \beta)})u\| \leq N_2^{(k;\alpha \beta)}(t)\|u\|$ for some continuous nonnegative functions $N_1^{(\alpha \beta)}$ and $N_2^{(k;\alpha \beta)}$, then
\begin{align}
 \left| \left< \left(U^{(k)}_t - U_t\right)^* \psi_1,\psi_2 \right> \right| 
 &\leq 2 \Big( \|u_1\| \|e(f_1)-e(f'_1)\| \|\psi_2\| + \|u_2\| \|e(f_2)-e(f'_2)\| \|\psi_1\| \Big) \nonumber \\
 &\quad + \frac{1}{k} \sum_{i=0}^\ell \left( M^{(k;\alpha_1(i)\alpha_2(i))}_1 \bigr(N_1^{(\alpha(i) \beta(i))}(t_{i+1}-t_i)  \right. \nonumber \\ 
 &\quad \left. + N_2^{(k;\alpha(i) \beta(i))}(t_{i+1}-t_i) \bigl)+ (t_{i+1}-t_i)M^{(k;\alpha_1(i)\alpha_2(i))}_2 \right) \|\psi_1'\|  \|\psi_2'\|.
 \label{eq:AE-Ut-err2b}
\end{align}
Moreover, it holds that
\begin{eqnarray}
 \lim_{k\rightarrow\infty} \left| \left< \left(U^{(k)}_t - U_t\right)^* \psi_1,\psi_2 \right> \right| &=& 0.
 \label{eq:AE-weak-convergence}
\end{eqnarray}
\end{corollary}
\begin{proof}
This proof follows similar arguments to the proof of Corollary \ref{cor:main2} using \eqref{eq:AE-Ut-err1} and \eqref{eq:AE-Ut-err1b} established in Lemma \ref{lemma:AE-main1}.
\end{proof}

\begin{theorem} \label{thm:AE-main3}
Suppose Assumptions \ref{assump:AE-scaling}, \ref{assump:AE-structure}, \ref{assump:AE-coefficients}, and \ref{assump:AE-boundedness} hold.
Let $0 < T < \infty$. For any $t \in [0,T]$, consider any $L'_t \in \mathbb{Z}_+$,
any $\psi'_t = \sum_{j=1}^{L'_t}  \psi'_{j,t}$,  with $\psi'_{j,t} = u'_{j,t} \otimes e(g'_{j,t}) \in {\cal U}^{(k)}$  and  $u'_{j,t} \neq 0$. Also, consider any $f' \in \mathfrak{S}'$. Let $\ell$ be a positive integer and $t_0 <t_1 < \ldots < t_{\ell} <t_{\ell +1}=t$ be a sequence such that $f'= \sum_{i=0}^{\ell} \alpha(i) \mathds{1}_{[t_i,t_{i+1})}$  and $g'_{j,t} = \sum_{i=0}^{\ell} \beta_{j,t}(i) \mathds{1}_{[t_i,t_{i+1})}$.  Let $u \in {\cal H}^{(k)}$ with $\| u\|=1$, and $|f\rangle  = e(f)/\|e(f)\| \in {\cal F}$  (i.e., $|f\rangle$ is a coherent state with amplitude $f$), and $\psi = u \otimes |f \rangle$. Then,
\begin{align}
 \left\|\left(U^{(k)}_t - U_t\right)^* \psi \right\|^2 &\leq 4 \Big(  \| |f \rangle - |f' \rangle \|   + \| {U_t}^* \psi - \psi'_t \| \Big) \nonumber \\
 &\quad + \frac{2}{k} \sum_{j=1}^{L'_t} \sum_{i=1}^\ell \left( 2M^{(k;\alpha(i)\beta_{j,t}(i))}_1 + (t_{i+1}-t_i)M^{(k;\alpha(i)\beta_{j,t}(i))}_2 \right) \|\psi'_{j,t}\|, 
 \label{eq:AE-Ut-err3}
\end{align}
If in addition,  $T_t^{(\alpha \beta)}$ and $T_t^{(k;\alpha \beta)}$ are also norm continuous for each $\alpha,\beta$, with 
 $\| (I-T_t^{(\alpha \beta)})u\| \leq N_1^{(\alpha \beta)}(t) \|u\|$ and $\| (I-T_t^{(k;\alpha \beta)})u\| \leq N_2^{(k;\alpha \beta)}(t)\|u\|$ for some continuous nonnegative functions $N_1^{(\alpha \beta)}$ and $N_2^{(k;\alpha \beta)}$, then
\begin{align}
 \left\|\left(U^{(k)}_t - U_t\right)^* \psi \right\|^2 &\leq 4 \Big(  \| |f \rangle - |f' \rangle \|   + \| {U_t}^* \psi - \psi'_t \| \Big) \nonumber \\
 &\quad + \frac{2}{k} \sum_{j=1}^{L'_t} \sum_{i=1}^\ell \left( M^{(k;\alpha(i)\beta_{j,t}(i))}_1 \bigl(N_1^{(\alpha(i) \beta(i))}(t_{i+1}-t_i)  \right. \notag \\
 &\quad \left. + N_2^{(k;\alpha(i) \beta(i))}(t_{i+1}-t_i)\bigr)+ (t_{i+1}-t_i)M^{(k;\alpha(i)\beta_{j,t}(i))}_2 \right) \|\psi'_{j,t}\|. 
 \label{eq:AE-Ut-err3b}
\end{align}
 Moreover, it holds that
\begin{align}
 \lim_{k\rightarrow\infty} \left\|\left(U^{(k)}_t - U_t\right)^* \psi \right\| = 0.
 \label{eq:AE-strong-convergence}
\end{align}
for any $t \in [0,T]$ with $0 < T < \infty$.
\end{theorem}
\begin{remark}
As with Theorem \ref{thm:main3}, a stronger  result of strong convergence uniformly over compact time intervals  $\lim_{k\rightarrow\infty} \mathop{\sup}_{0 \leq t \leq T} \left\|\left(U^{(k)}_t - U_t\right)^* \psi \right\| =0$ has been established in \cite[Theorem 11]{BvHS08} for adiabatic elimination  based on a Trotter-Kato theorem, but without error bounds for finite values of $k$.
\end{remark}
\begin{proof}
This proof follows similar arguments to the proof of Theorem \ref{thm:main3} using \eqref{eq:AE-Ut-err2} and \eqref{eq:AE-Ut-err2b} established in Corollary \ref{cor:AE-main2}.
\end{proof}


 \subsection{Adiabatic elimination examples}

\subsubsection{Elimination of a harmonic oscillator:}

Consider a class of open quantum systems that comprises a finite-dimensional atomic system coupled to a harmonic oscillator which is driven by $m$ external coherent fields (originally presented in \cite{BvHS08}).
Let ${\cal H} = {\cal H}'\otimes \ell^2$, where ${\cal H}'$ is a finite-dimensional Hilbert space and $\ell^2$ is the space of infinite complex-valued sequences with $\sum_{n=1}^\infty |x_n|^2 < \infty$. Similar to Example \ref{ex:mabuchi}, let $\{\ket{n}\}_{n\geq 0}$ be an orthonormal Fock state basis of $\ell^2$. On this basis, the annihilation, creation, and number operators can be defined (see, e.g., \cite{BvHS08}) satisfying 
\begin{align*}
 a\ket{n} = \sqrt{n}\ket{n-1}, \qquad a^*\ket{n} = \sqrt{n+1}\ket{n+1}, \qquad a^*a\ket{n} = n\ket{n},
\end{align*}
respectively. Following \cite{BvHS08}, we choose the dense domain ${\cal D} = {\cal H}'\otimes {\rm span}\{ \ket{n} \mid n \in \mathbb{Z}_+ \}$.
Let us define $K^{(k)} = \imath H^{(k)} - \frac{1}{2}\sum_{i=1}^{m}({L^{(k)}_i}^*  L^{(k)}_i)$.
Consider the system operators $(S^{(k)},L^{(k)},H^{(k)})$ defined such that
\begin{align*}
 K^{(k)} &= k^2 E_{11}\otimes a^*a + k (E_{10}\otimes a^* + E_{01}\otimes a) + E_{00}\otimes I, \\
 {L^{(k)}_j}^* &= kF_j \otimes a^* + G_j\otimes I, \\
 {S^{(k)}_{ji}}^* &= W_{ij}\otimes I
\end{align*}
where $E_{11},E_{10},E_{01},E_{00},F_j,G_j,W_{ij}$ are bounded operators on the finite-dimensional space ${\cal H}'$.

Now consider ${\cal H}_0 = {\cal D}_0 = {\cal H}'\otimes \mathbb{C}\ket{0}$. That is, the harmonic oscillator is eliminated from the model as it is forced into its ground state (i.e., $\ket{0}$) in the limit as $k \rightarrow \infty$. In quantum optics, this process is the adiabatic elimination of an optical cavity in the strong damping limit. 
Now consider an approximation system with the operators $(S,L,H)$ which are defined such that 
\begin{align*}
 S_{ji}^* &= \sum_{\ell=1}^m W_{i\ell} \left( F_\ell^* (E_{11})^{-1} F_j + \delta_{\ell j} \right) \otimes I. \\
 L_j^* &= \left(G_j - E_{01}(E_{11})^{-1}F_j\right) \otimes I, \\
 H &= \Im\left\{ E_{00} - E_{01}(E_{11})^{-1} E_{10}\right\}  \otimes I.
\end{align*}
Here, we stress that $S$ is unitary and $H$ is self-adjoint \cite[Proposition 14]{BvHS08}.

It has been shown in \cite{BvHS08} that Conditions \ref{cond:cocycle} and \ref{cond:core} hold for the above systems.
The original system satisfies Assumption \ref{assump:AE-scaling} with $Y = E_{11}\otimes a^*a$, $A = E_{10}\otimes a^* + E_{01}\otimes a$ and $B = E_{00}\otimes I$.
Suppose $E_{11}$ has a bounded inverse, then Assumption \ref{assump:AE-structure} is satisfied with $\tilde{Y}$ which is defined such that
$\tilde{Y}\psi\otimes\ket{n} = \frac{1}{n} (E_{11})^{-1} \psi\otimes\ket{n}$ for $n \geq 1$ and $\psi \in {\cal H}'$.
The operators $(S,L,H)$ also satisfy Assumption \ref{assump:AE-coefficients}. It now remains to show that Assumption \ref{assump:AE-boundedness} holds.

Note that $P_{{\cal H}_0^\perp} B^{(\alpha\beta) }P_{{\cal H}_0} = 0$.
Now let $Q^{(\beta)} = E_{10} + \sum_{j=1}^m F_j \beta_j $ and let $P^{(\alpha\beta)} = -\frac{|\alpha|^2+|\beta|^2}{2} + E_{00} + \sum_{j=1}^m G_j\beta_j + \sum_{i,j=1}^m \alpha_i^* W_{ij} \left(\beta_j - G_j^*\right)$.
Note that $Q^{(\beta)}$ and $P^{(\alpha\beta)}$ are bounded operators on the finite-dimensional Hilbert space ${\cal H}'$ for any $\alpha,\beta \in \mathbb{C}^m$ since $E_{10}$, $E_{00}$, $F_j$, $G_j$ and $W_{ij}$ are bounded operator on ${\cal H}'$.
For any $\psi \in {\cal H}'$, we have that 
\begin{align*}
 (\tilde{Y}A^{(\alpha\beta)}) \psi \otimes \ket{0} 
 &= (E_{11})^{-1}Q^{(\beta)} \psi \otimes \ket{1}, \\
 (\tilde{Y}P_{{\cal H}_0^\perp}A^{(\alpha\beta)}\tilde{Y}A^{(\alpha\beta)}) \psi \otimes \ket{0} 
 &= \frac{1}{\sqrt{2}}(E_{11})^{-1}Q^{(\beta)} \psi \otimes \ket{2}, \\
 (B^{(\alpha\beta)}\tilde{Y}A^{(\alpha\beta)}) \psi \otimes \ket{0} 
 &= P^{(\alpha\beta)} (E_{11})^{-1} Q^{(\beta)} \psi \otimes \ket{1}, \\
 (A^{(\alpha\beta)}\tilde{Y}P_{{\cal H}_0^\perp}A^{(\alpha\beta)}\tilde{Y}A^{(\alpha\beta)}) \psi \otimes \ket{0} 
 &= \sqrt{\frac{3}{2}}Q^{(\beta)}(E_{11})^{-1}Q^{(\beta)} \psi \otimes \ket{3} \\
 &\quad + \left( E_{01} - \sum_{i,j=1}^m \alpha_i ^* W_{ij} F_j^* \right)(E_{11})^{-1}Q^{(\beta)} \psi \otimes \ket{1}, \\
 (B^{(\alpha\beta)}\tilde{Y}P_{{\cal H}_0^\perp}A^{(\alpha\beta)}\tilde{Y}A^{(\alpha\beta)}) \psi \otimes \ket{0} 
 &=  \frac{1}{\sqrt{2}}P^{(\alpha\beta)} (E_{11})^{-1}Q^{(\beta)} \psi \otimes \ket{2}.
\end{align*}
From these identities, we see that Assumption \ref{assump:AE-boundedness} holds because $E_{11}$ has a bounded inverse, $E_{01}$, $F_j$, $W_{ij}$ $Q^{(\beta)}$, $P^{(\alpha\beta)}$ are bounded operators, and ${\cal L}^{(\alpha\beta)}$ is defined on the finite-dimensional subspace ${\cal H}_0$ (i.e., it is also a bounded operator). Thus, the conditions of Lemma \ref{lemma:AE-semigroup}, Lemma \ref{lemma:AE-main1}, Corollary \ref{cor:AE-main2}, and Theorem \ref{thm:AE-main3} have been verified.

\subsubsection{Atom-cavity model \cite[Example 15]{BvHS08}:}
Consider a system consisting of a three-level atom coupled to an optical cavity. The cavity and the uncoupled leg of the atom is driven by an external coherent field ($m=1$).
Let ${\cal H} = \mathbb{C}^3 \otimes \ell^2$. As in the previous example, we consider the orthonormal Fock state basis $\{ \ket{n} \}_{n\geq 0}$ of $\ell^2$. We will use $\ket{e} = (1,0,0)^\top$, $\ket{+} = (0,1,0)^\top$, and $\ket{-} = (0,0,1)^\top$ to denote the canonical basis vectors in $\mathbb{C}^3$. In the basis $\{ \ket{e},\ket{+},\ket{-} \}$, let us define $\sigma_{+}^{(+)} = \ket{e}\bra{+}$ and $\sigma_{+}^{(-)} = \ket{e}\bra{-}$.
We also define $\sigma_{-}^{(\pm)} = {\sigma_{+}^{(\pm)}}^*$ and $P_{-} = \ket{-}\bra{-}$. 
Here, ${\cal D} = \mathbb{C}^3 \otimes {\rm span}\{ \ket{n} \mid n \in \mathbb{Z}_+ \}$.
Under the rotating wave approximation and in the rotating frame of reference, the system is described by the following operators: \cite{BvHS08}
\begin{align*}
 S^{(k)} &= I, \\ 
 L^{(k)} &= I \otimes k\sqrt{\gamma} a, \\
 H^{(k)} &= \imath gk^2\left( \sigma_{+}^{(+)}\otimes a - \sigma_{-}^{(+)}\otimes a^* \right) + \imath k\left( \sigma_{+}^{(-)}\alpha - \sigma_{-}^{(-)} \alpha^* \right)\otimes I
\end{align*}
where $\gamma,g > 0$. Here, $\alpha \in \mathbb{C}$ is the amplitude of the external coherent field driving the cavity and the uncoupled leg of the atom.

Now consider ${\cal H}_0 = {\cal D}_0 = {\rm span}\{ \ket{+}\otimes\ket{0}, \ket{-}\otimes\ket{0} \}$. 
That is, the cavity oscillator and the excited state of the atom (i.e., $\ket{e}$) is eliminated from the model in the limit as $k \rightarrow \infty$. 
Consider an approximating system described the operators $(S,L,H)$ which are defined as
\begin{align*}
 S = \left(I - 2P_{-}\right) \otimes I, \qquad L = -\frac{\alpha\sqrt{\gamma}}{g} \sigma_{-}^{(+)}\sigma_{+}^{(-)}  \otimes I, \qquad H = 0.
\end{align*}
It can be easily verified that $S$ is unitary and $H$ is self-adjoint.

Again, it has been shown in \cite{BvHS08} the above systems satisfy Conditions \ref{cond:cocycle} and \ref{cond:core}. 
We then see that Assumption \ref{assump:AE-scaling} with $Y = -\frac{\gamma}{2}\otimes a^*a + g\left(\sigma_{-}^{+}\otimes a^* - \sigma_{+}^{(+)} \otimes a\right)$, $A = \left( \sigma_{-}^{(-)}\alpha^* - \sigma_{+}^{(-)}\alpha \right)\otimes I$, $B = 0$, $F = \sqrt{\gamma} \otimes a^*$, $G = 0$, and $W = I$. 
Let us define
\begin{align*}
 {\cal H}_j = {\rm span}\{ \ket{+}\otimes\ket{j}, \ket{-}\otimes\ket{j}, \ket{e}\otimes\ket{j-1} \}, \qquad \textnormal{for } j \in \mathbb{Z}_+.
\end{align*}
Assumption \ref{assump:AE-structure} holds with $\tilde{Y}$ which is defined, with respect to the basis $\{\ket{+}\otimes\ket{j}, \ket{-}\otimes\ket{j}, \ket{e}\otimes\ket{j-1}\}$, as \cite{BvHS08}
\begin{align*}
 \tilde{Y}|_{{\cal H}_j} = -\frac{1}{d_j}\left[\begin{array}{ccc} 
 \frac{\gamma}{2}(j-1) & 0 & g\sqrt{j} \\ 0 & \frac{2d_j}{j\gamma} & 0 \\ -g\sqrt{j} & 0 & \frac{j\gamma}{2}
 \end{array}\right], \qquad
 d_j = \frac{j(j-1)\gamma^2}{4} + jg^2.
\end{align*}
It can also be seen that Assumption \ref{assump:AE-coefficients} holds for the defined $(S,L,H)$. Thus, it remains to show that Assumption \ref{assump:AE-boundedness} holds.
Note that $B^{(\alpha\beta)} = -\frac{1}{2}|\alpha-\beta|^2$ and thus, $P_{{\cal H}_0^\perp} B^{(\alpha\beta)} P_{{\cal H}_0} = 0$.
Then note that 
\begin{align*}
 \tilde{Y}A^{(\alpha\beta)}P_{{\cal H}_0} &\subset {\cal H}_1, \\
 \tilde{Y}P_{{\cal H}_0^\perp}A^{(\alpha\beta)}\tilde{Y}A^{(\alpha\beta)}P_{{\cal H}_0} &\subset {\cal H}_2, \\
 B^{(\alpha\beta)}\tilde{Y}A^{(\alpha\beta)}P_{{\cal H}_0} &\subset {\cal H}_1, \\
 A^{(\alpha\beta)}\tilde{Y}P_{{\cal H}_0^\perp}A^{(\alpha\beta)}\tilde{Y}A^{(\alpha\beta)}P_{{\cal H}_0} &\subset {\cal H}_3 \oplus {\cal H}_1, \\
 B^{(\alpha\beta)}\tilde{Y}P_{{\cal H}_0^\perp}A^{(\alpha\beta)}\tilde{Y}A^{(\alpha\beta)}P_{{\cal H}_0} &\subset {\cal H}_2.
\end{align*}
From the above relations and the fact ${\cal L}^{(\alpha\beta)}$ is defined on the finite-dimensional subspace ${\cal H}_0$ (i.e., it is a bounded operator), we have that Assumption \ref{assump:AE-boundedness} holds.
\paragraph{Numerical example of the atom-cavity model:} 
Consider $\gamma = 25$, $g = 5$, and $t \in [0,T]$ with $T = 1$.
Let $\psi = (\ket{-}\otimes\ket{0})\otimes | \alpha\mathds{1}_{[0,T]}\rangle$ with $\alpha = 0.1$.
We will now compute the error bound on $\left|\left| \left(U_t - U^{(k)}_t \right)^*\psi \right|\right|$ for different values of $k$.
From \eqref{eq:AE-Ut-err3} (established in Theorem \ref{thm:AE-main3}) and the fact that $\alpha\mathds{1}_{[0,T]} \in \mathfrak{S}'$ is a simple function, we have that
\begin{align*}
 \left\|\left(U_t - U^{(k)}_t\right)^* \psi \right\|^2 
  &\leq 4  \left(\left\| {U_t}^* \psi - \psi'_t \right\|\right) \nonumber \\
  &\quad + \frac{2}{k} \sum_{j=1}^{L'_t} \sum_{i=1}^\ell \left( 2M^{(k;\alpha(i)\beta_{j,t}(i))}_1 + (t_{i+1}-t_i)M^{(k;\alpha(i)\beta_{j,t}(i))}_2 \right) \|\psi'_{j,t}\|. 
\end{align*}
where $\psi'_t = \sum_{j=1}^{L'_t} \psi'_{j,t}$, and $\psi'_{j,t} = u'_{j,t}\otimes e(g'_{j,t})$ with $g'_{j,t} \in \mathfrak{S}'$.
Similar to Example \ref{ex:mabuchi}, $\left\| {U_t}^* \psi - \psi'_t \right\|$ is bounded by \eqref{eq:sim-bound}.
Again, to find an appropriate $\psi'_t$, let us define the cost function $J(\psi'_t)$ as in \eqref{eq:sim-cost}
for any $\psi'_t  = \sum_{j=1}^{L_t} u'_{j,t}\otimes e(g'_{j,t}) \in {\cal U}^{(k)}$. We then find $\psi'_t$  that is a local minimizer of $J$ by numerical optimization.

For computational simplicity, we fix $L_t = 5$ and set $t_{i+1}-t_i = 10^{-3}$ for all $i= 0,1,2,\ldots,\ell$. 
With $t = T$, we then have that $\ell = 999$. Note that simultaneously optimizing over $1000$ time intervals is computationally intensive. Thus, we simplify the computation further by optimizing sequentially over blocks of $10$ time intervals at a time. Thus optimization is done over 100 blocks. Optimization of the first block is initialized with $\psi_{j,t}' = (\ket{-}\otimes\ket{0})\otimes e( \alpha\mathds{1}_{[0,t_{10}]} )$ for all $j = 1,2,\ldots,L_t$. The optimization result of each block is then used to initialize the optimization of the next block in the sequence. 
A local minimizer $\psi_{j,t}' = u_{j,t}' \otimes e(g'_{j,t})$, for $j = 1,2,\ldots,L_t$, was found  using Matlab as before, given by 
$u_{1,t}' = (0.1008+0.1549\imath) (\ket{+}\otimes\ket{0}) + (0.2311+0.1012\imath) (\ket{-}\otimes\ket{0})$, 
$u_{2,t}' = (0.0949-0.0031\imath) (\ket{+}\otimes\ket{0}) + (0.1500+0.0270\imath) (\ket{-}\otimes\ket{0})$, 
$u_{3,t}' = (-0.0245+0020\imath) (\ket{+}\otimes\ket{0}) + (0.2012+0.0311\imath) (\ket{-}\otimes\ket{0})$, 
$u_{4,t}' = (0.3266-0.1431\imath) (\ket{+}\otimes\ket{0}) + (0.2928-0.0995\imath) (\ket{-}\otimes\ket{0})$,  
$u_{5,t}' = -(0.4956+0.0101\imath) (\ket{+}\otimes\ket{0}) + (0.1232-0.0591\imath) (\ket{-}\otimes\ket{0})$.
The overall resulting optimized cost is $J( \psi'_t) = 0.0046$. Using $\psi'_t$, the error bounds on $\left|\left| \left(U_t - U^{(k)}_t \right)^*\psi \right|\right|$ using \eqref{eq:AE-Ut-err2} for various values of $k$ are shown in Table \ref{tab:ae-results}. 
\begin{table}[!t]
 \centering
  \caption{Numerical computation of error bounds on $\left|\left| \left(U_t - U^{(k)}_t \right)^*\psi \right|\right|$ for the adiabatic elimination approximation in the atom-cavity example}
 \begin{tabular}{c||c}
  k & Error bound \\ \hline\hline
  $10^4$ & $0.9347$ \\
  $10^5$ & $0.2957$ \\
  $10^6$ & $0.0309$ \\
  $10^7$ & $0.0131$ \\
  $10^8$ & $0.0096$
 \end{tabular}
 \label{tab:ae-results}
\end{table}


\section{Conclusion} \label{sec:conclusion}

This work has developed a framework for developing error bounds for finite dimensional  approximations of input-output quantum stochastic models defined on infinite-dimensional underlying Hilbert spaces, with possibly unbounded coefficients in their QSDEs. The framework exploits a contractive semigroup that can be associated with the QDESs. This gives for the first time error bound expressions for two types of approximations that are often employed in the literature, subspace truncation and adiabatic elimination. The bounds  are in principle computable and vanish for each $t$ in the limit as the parameter $k$, representing the dimension of the approximating subspace in the case of subspace truncation approximation or a large scaling parameter in the case of adiabatic elimination, goes to $\infty$. The theory developed was applied to some physical examples taken from the literature.

There are several directions for further investigation along the theme initiated in this paper. Devising a more efficient method for computing bounds for the term $\| {U_t^{(k)}}^* \psi - \sum_j \psi'_{j,t}\|$ for subspace truncation and $\| {U_t}^* \psi - \sum_j \psi'_{j,t}\|$ for adiabatic elimination, beyond the computationally intensive optimization based approach that was considered herein, will be important. Tensor network methods that have recently met a lot of success  for efficient simulation of one dimensional many-body systems could potentially be important for this purpose. There also remains the question of the conservatism in the error bounds and if there could be tighter bounds that can be achieved by using a different set of assumptions. In the numerical example of adiabatic elimination, the bound  \eqref{eq:AE-Ut-err3} was employed rather than the potentially less conservative \eqref{eq:AE-Ut-err3b}. This is because the latter requires determining whether  $T_t^{(k;\alpha \beta)}$ is a norm continuous semigroup and finding the bounding function $N_2^{(k;\alpha \beta)}$, a non-trivial task in general that deserves further investigation. Moreover, it would be interesting to see if there are exactly solvable QSDE models of a physical system with an infinite-dimensional system Hilbert space for initial states that lie in a finite-dimensional subspace of the original, against which the conservatism of error bounds can be assessed. The authors are  currently unaware of any such exactly solvable models.

\section*{Acknowledgements}
The authors are grateful for the support of the Australian Research Council under Discovery Project DP130104191. 
\section*{References}

\small
\bibliographystyle{IEEEtran}
\bibliography{IEEEabrv,rip,otr}

\end{document}